\RequirePackage{fix-cm}

\documentclass[smallextended]{svjour3}       
\smartqed  

\usepackage{graphicx}
\usepackage{diagbox}
\usepackage{algpseudocode}
\usepackage{algorithm}
\usepackage{tikz}
\usepackage{url}
\usepackage{subfigure, graphicx}
\usepackage{float}
\usepackage{amsmath,amsfonts} 
\usepackage{mathrsfs} 
\usepackage{amssymb}
\usepackage{wasysym}
\usepackage{multirow}
\usepackage{enumerate}
\usepackage{enumitem}
\usetikzlibrary{shapes}
\usetikzlibrary{arrows,automata}
\usetikzlibrary{positioning}
\usetikzlibrary{patterns}
\usetikzlibrary{backgrounds}
\usepackage{array}

\newcommand{\Hi}[1]{$\mathscr{HI}$}
\newcommand{\Lo}[1]{$\mathscr{LO}$}

\newcommand{\my}[1]{{\color{blue}{#1}}}

\newcommand\myeq{\mathrel{\overset{\makebox[0pt]{\mbox{\normalfont\tiny\sffamily def}}}{=}}}

\allowdisplaybreaks
\begin{document}

\title{Multi-Rate Fluid Scheduling of Mixed-Criticality Systems on Multiprocessors
\thanks{This research was funded in part by the Ministry of Education, Singapore, Tier-1 grant RG21/13 and Tier-2 grant ARC9/14, and by the Start-Up-Grant from SCSE, NTU, Singapore.  This research was also partly supported by Basic Science Research Program of the National Research Foundation of Korea (NRF- 2015R1D1A1A01057018).}
}

\author{Saravanan Ramanathan \and Arvind Easwaran \and Hyeonjoong Cho}


\institute{Saravanan Ramanathan \and Arvind Easwaran\at
              School of Computer Science and Engineering, Nanyang Technological University, Singapore\\
              \email{\{saravana016,arvinde\}@e.ntu.edu.sg}
           \and
           Hyeonjoong Cho\at
           Department of Computer and Information Science, Korea University\\
           \email{raycho@korea.ac.kr}         
}


\maketitle

\begin{abstract}
In this paper we consider the problem of mixed-criticality (MC) scheduling of implicit-deadline sporadic task systems on a homogenous multiprocessor platform. Focusing on dual-criticality systems, algorithms based on the fluid scheduling model have been proposed in the past. These algorithms use a dual-rate execution model for each high-criticality task depending on the system mode. Once the system switches to the high-criticality mode, the execution rates of such tasks are increased to meet their increased demand. Although these algorithms are speed-up optimal, they are unable to schedule several feasible dual-criticality task systems. This is because a single fixed execution rate for each high-criticality task after the mode switch is not efficient to handle the high variability in demand during the transition period immediately following the mode switch. This demand variability exists as long as the carry-over jobs of high-criticality tasks, that is jobs released before the mode switch, have not completed. Addressing this shortcoming, we propose a multi-rate fluid execution model for dual-criticality task systems in this paper. Under this model, high-criticality tasks are allocated varying execution rates in the transition period after the mode switch to efficiently handle the demand variability. We derive a sufficient schedulability test for the proposed model and show its dominance over the dual-rate fluid execution model. Further, we also present a speed-up optimal rate assignment strategy for the multi-rate model, and experimentally show that the proposed model outperforms all the existing MC scheduling algorithms with known speed-up bounds.
\end{abstract}

\section{Introduction}
\label{sec:introduction}

The mixed-criticality (MC) model proposed by Vestal \cite{vestal} has received a lot of attention in the literature on real-time scheduling. Several studies exist on the design of multiprocessor MC scheduling algorithms; see~\cite{burns_review} for a review. To evaluate the schedulability performance of these algorithms two techniques are generally used: 1) \emph{experimental evaluation} in which the schedulability is assessed over a wide variety of task systems, and 2) analytical performance bounds such as the \emph{speed-up bound}\footnote{The speed-up bound of a scheduling algorithm is defined as the maximum additional processor speed required to schedule any feasible task system using the algorithm~\cite{kalyanasundaram_speedup}.} are derived. Scheduling algorithms with high schedulability in experimental evaluation as well as good (low) speed-up bound are highly desirable. 

Based on the above motivation, fluid scheduling algorithms have recently been proposed for MC scheduling of task systems on homogenous multiprocessor platforms~\cite{baruah_mcf,lee_mcfluid}. For scheduling implicit-deadline MC task systems with two criticality levels (dual-criticality systems), Lee et al. \cite{lee_mcfluid} proposed a dual-rate fluid scheduling model in which each high-criticality task executes using two rates and each low-criticality task executes using a single rate. The high-criticality task rates depend on the mode in which the system is operating; they execute using higher rates once the system switches to the high-criticality mode. Note low-criticality tasks are suspended after the system switches to the high-criticality mode. A rate assignment strategy called MC-Fluid has also been proposed, and it is shown to be \emph{dual-rate optimal}; if there is a feasible dual-rate assignment for a dual-criticality task system, then MC-Fluid is guaranteed to find it. Subsequently, Baruah et al. \cite{baruah_mcf} derived a simplified rate assignment strategy called MCF for dual-rate fluid scheduling, and showed that both MC-Fluid and MCF have an optimal speed-up bound of $4/3$.

Although algorithms such as MC-Fluid are dual-rate and speed-up optimal, there are several feasible dual-criticality task systems that are not schedulable under the dual-rate fluid scheduling model. In fact, in experimental evaluations it has been observed that the performance of MC-Fluid and MCF drop significantly for task systems with utilization greater than $3/4$~\cite{ramanathan_wmc}. The primary reason for this under-performance is the fact that a single fixed execution rate for each high-criticality task after the mode switch is not efficient to handle the high variance in demand during the transition period immediately following the mode switch. This demand variance arises due to the carry-over jobs of high-criticality tasks, that is jobs released before the mode switch but incomplete at the time of mode switch. Since the execution of such jobs is affected by low-criticality tasks, their left-over demand after mode switch may be different (usually higher) than the demand of pure high-criticality jobs that are released after the mode switch. Therefore, in the transition period while such carry-over jobs are executing, the overall demand of high-criticality tasks is varying depending on the number of pending carry-over jobs.

To address the aforementioned limitation of the dual-rate fluid scheduling model, in this paper we propose a \emph{multi-rate fluid scheduling model} for dual-criticality task systems scheduled on homogenous multiprocessor platforms. Under this model, each high-criticality task executes using a single execution rate in the low-criticality mode and a set of execution rates in the high-criticality mode. In particular, the task executes using different rates over time until all carry-over jobs are guaranteed to be completed, and the pure high-criticality jobs released subsequently execute using a single rate based on task utilization. These rates and the time duration for which they are applicable are all determined offline so as to ensure worst-case schedulability. Similar to the dual-rate model, each low-criticality task uses a single execution rate based on task utilization before mode switch and is suspended after the mode switch. Thus, by using a fine-grained rate allocation in the transition period, the multi-rate model is able to accommodate a higher left-over demand for the carry-over jobs after the mode switch. This in turn enables it to accommodate a higher demand in the low-criticality mode, thus improving schedulability over the dual-rate model. We now illustrate this benefit of the multi-rate model using a simple example.

\begin{example}
In Table~\ref{tab:example} we show a dual-criticality task system that is schedulable under the multi-rate model on $2$ processors, but is not schedulable under the dual-rate model. Under the dual-rate model, both the carry-over job and the pure high-criticality jobs of a high-criticality task execute using a single rate. As can be seen, using the dual-rate model, the resulting execution rates in low-criticality mode are not feasible ($>2$). Since these rates have been obtained using the dual-rate optimal MC-Fluid algorithm, we can conclude that the task system is not schedulable under the dual-rate model. One may wonder if a slight modification of the dual-rate model (two execution rates after the mode switch), in which only the carry-over job uses the execution rate determined by MC-Fluid and the pure high-criticality jobs use another execution rate based on task utilization, does any better. However, this does not help to improve schedulability, because the execution rates of carry-over jobs in fact determine the rate allocation for the low-criticality mode. In the case of the multi-rate model, each high-criticality task is assigned multiple execution rates in the high-criticality mode as shown in the table. Immediately after entering the high-criticality mode in window $w_1$, while tasks $\tau_1$ and $\tau_2$ execute at the maximum rate of $1$, $\tau_3$ does not execute at all. Later on in window $w_3$, when $\tau_1$ and $\tau_2$ are executing at a much lower rate equal to their utilization, $\tau_3$ is able to execute at a rate of $0.5$. This is even higher than the rate it was allocated in the dual-rate model ($0.36$). Thus, the average rate allocated to $\tau_3$ across windows $w_1, w_2$ and $w_3$ in the multi-rate model is higher than the rate allocated to it in the dual-rate model. As a result its rate in low-criticality mode is reduced, and the overall task system becomes schedulable (schedulability is formally verified in Sect.~\ref{sec:test}).
\end{example}

\begin{table}[]
	\centering
	\setlength{\extrarowheight}{2.5pt}
	\resizebox{\textwidth}{!}{
	\begin{tabular}{|l|l|l|l|l|l|l|l|l|l|l|l|l|l|l|l|l|}
		\hline
		Tasks & & & & & & \multicolumn{3}{c|}{Dual-rate/3-rate assignment} & \multicolumn{8}{c|}{Multi-rate assignment}\\ \cline{7-17}
		& $C_i^L$ & $C_i^H$ & $T_i$ & $u_i^L$ & $u_i^H$ & $\theta_i^L$ & $\theta_i^H(=\theta_i^C)$
		& $\theta_i^P(=u_i^H)$ & $\theta_i^L$ & $\theta_{i,1}^{H}$ & $w_1$ & $\theta_{i,2}^H$ & $w_2$ & $\theta_{i,3}^H$ & $w_3$ & $\theta_i^{H}(=u_i^H)$ \\ \hline
		\begin{tabular}[c]{@{}l@{}}$\tau_1$\end{tabular} 
		& 2.8  & 4.9 & 7 & 0.4 & 0.7 & 0.700 & 0.700 & 0.700 & 0.571428 & 1.000 & 2.1 & 0.7 & 0.4 & 0.7 & 13.76 & 0.7 \\ \hline
		\begin{tabular}[c]{@{}l@{}}$\tau_2$\end{tabular} 
		& 1.5  & 4 & 5 & 0.3 & 0.8 & 0.641 & 0.939 & 0.800 & 0.600 & 1.000 & 2.1 & 1.000  & 0.4 & 0.8  & 13.76 & 0.8\\ \hline
		\begin{tabular}[c]{@{}l@{}}$\tau_3$\end{tabular} 
		& 3.5  & 10.5 & 35 & 0.1 & 0.3 & 0.224 & 0.360 & 0.300 & 0.186766 & 0.000 & 2.1 & 0.300 & 0.4 & 0.500  & 13.76 & 0.3 \\ \hline
		\begin{tabular}[c]{@{}l@{}}$\tau_4$\end{tabular} 
		& 15.75  & - & 35 & 0.45 & - & 0.450 & - & \my{-} & 0.450 & - & - & - & -& - & - & -\\ \hline
		\begin{tabular}[c]{@{}l@{}}$\sum$\end{tabular} 
		&   &  &  &  &  & \textcolor{red}{2.015} & 1.999 & 1.80 & \textcolor{red!05!green!50!blue}{1.808294} & 2.000 & & 2.000 & & 2.000 & & 1.80 \\ \hline
	\end{tabular}}
	\caption{\label{tab:example}\textbf{Example task system schedulable under the multi-rate fluid model, but not under the dual-rate fluid model on $2$ processors}. $C_i^L$ denotes low-criticality execution time estimate, $C_i^H$ denotes high-criticality execution time estimate, $T_i$ denotes task period, $u_i^L=C_i^L/T_i$ and $u_i^H=C_i^H/T_i$. For the dual-rate fluid model, $\theta_i^L$ denotes execution rate in the low-criticality mode and $\theta_i^H$ denotes execution rate in the high-criticality mode.
	For the modified dual-rate fluid model (3-rate assignment), $\theta_i^L$ denotes execution rate in the low-criticality mode, $\theta_i^C$ denotes execution rate in the transition period and $\theta_i^P$ denotes execution rate in the pure high-criticality mode.
	For the multi-rate fluid model, $\theta_i^L$ denotes execution rate in the low-criticality mode, $\theta_{i,j}^H (j \in \{ 1,2,3 \})$ and $\theta_i^H$ denote several execution rates for the high-criticality mode, and $w_j (j \in \{1, 2, 3\})$ denotes duration of time for which rate $\theta_{i,j}^H$ will be used by task $\tau_i$. After $w_1+w_2+w_3$ time units in the high-criticality mode, each high-criticality task $\tau_i$ will use the rate $\theta_i^H$.}
\end{table}
\paragraph{Contributions.} The contributions of this paper can be summarized as follows.
\begin{itemize}
\item We propose a new multi-rate fluid model for scheduling implicit-deadline MC task systems on a homogenous multiprocessor platform (Sect.~\ref{sec:model}).
\item We derive a sufficient schedulability test for the multi-rate model (Sect.~\ref{sec:test}), and show that it dominates the schedulability test for the dual-rate model (Sect.~\ref{sec:property}).
\item We present a convex optimization based rate and window duration assignment strategy for the multi-rate model called SOMA (Speed-up Optimal Multi-rate Assignment), and prove that it is speed-up optimal with a speed-up bound of $4/3$ (Sect.~\ref{sec:rates}).
\item We present results from extensive experimental evaluation and show that SOMA outperforms all the other multiprocessor MC scheduling algorithms with known speed-up bounds.
\end{itemize}

\paragraph{Related Work.} Several studies have been done on the design of multi-core MC scheduling algorithms in recent years~(\cite{anderson,baruah_multicoremc,baruah_mcf,gu_mpvd,lee_mcfluid,li_globalmc,pathan,ren_taskgrouping,rodriguez}). Of which, only a few provided both experimental evaluation and analytical performance bounds~(\cite{baruah_multicoremc,baruah_mcf,lee_mcfluid,li_globalmc}).
Li and Baruah \cite{li_globalmc} proposed GLO-EDF\textunderscore VD, a global scheduling algorithm combining the multiprocessor fixed priority algorithm fpEDF and uniprocessor virtual deadline based MC algorithm EDF-VD, and proved that the algorithm has a speed-up bound of $\sqrt{5}+1$. Baruah et al. \cite{baruah_multicoremc} presented a EDF-VD based partitioned scheduling algorithm PAR-EDF\textunderscore VD and proved that the algorithm has a speed-up bound of $8/3 - 4/3m$. They also showed through experimental evaluation that the partitioned algorithm offers better schedulability than the global variant. Lee et al. \cite{lee_mcfluid} proposed the dual-rate fluid scheduling model and rate assignment algorithm called MC-Fluid, and showed that MC-Fluid has a speed-up bound of $(\sqrt{5}+1)/2$. Recently, Baruah et al. \cite{baruah_mcf} proposed a simplified rate assignment algorithm for the dual-rate fluid model called MCF, and proved that both MC-Fluid and MCF are speed-up optimal with a speed-up bound of $4/3$.

In this work, we focus on fluid scheduling model with multiple rates for dual-criticality systems to improve their schedulability with an optimal speed-up bound of $4/3$.

\paragraph{Organization.} The remainder of the paper is as follows. We describe the system model, notations and the dual-rate fluid scheduling algorithm in Sect.~\ref{sec:system_model}. We introduce the multi-rate fluid model in Sect.~\ref{sec:model} and prove its correctness in Sect.~\ref{sec:test}.  We present the properties of the multi-rate fluid model in Sect.~\ref{sec:property}. In Sect.~\ref{sec:rates} we present the execution rates assignment strategy and a heuristic for computing these execution rates. We describe the experiments conducted to evaluate the performance of our algorithm with the existing algorithms in Sect.~\ref{sec:experiments}. Sect.~\ref{sec:summary} concludes the paper and presents the possible future work in fluid scheduling.
\section{System Model}
\label{sec:system_model}

We consider an implicit-deadline sporadic MC task system with two criticality levels (LO and HI). Each MC task $\tau_i$ is characterized by a tuple $(T_i,\chi_i,C_i^L,C_i^H)$, where
\begin{itemize}
\item $T_i \in\mathbb{R}^+$ is the minimum release separation time; we assume an implicit-deadline task model, where deadline $D_i$ of a task is equal to $T_i$.
\item $\chi_i \in\{LO,HI\}$ denotes the criticality level of the task; we use the shortcut notation \emph{LO-task} and \emph{HI-task} to denote a LO-criticality and HI-criticality task respectively.
\item $C_i^L \in\mathbb{R}^+$ is the LO-criticality worst case execution time (WCET) value (denoted as \emph{LO-WCET}).
\item $C_i^H \in\mathbb{R}^+$ is the HI-criticality WCET value (denoted as \emph{HI-WCET}); we assume $C_i^L \leq C_i^H$ for all HI-tasks and $C_i^L = C_i^H$ for all LO-tasks.
\end{itemize}

We consider a MC task system $\tau$ comprised of $n$ sporadic tasks $\{ \tau_1, \ldots , \tau_n \}$. Let $n_H$ denote the number of HI-tasks in the system. Let $\tau_L\myeq\{ \tau_i\in\tau\mid \chi_i=LO\}$ and $\tau_H\myeq\{\tau_i\in\tau\mid \chi_i=HI\}$ denote the set of LO and HI-tasks respectively.

Task and system utilizations are denoted as follows. $u_i^L\myeq C_i^L/T_i$, $u_i^H\myeq C_i^H/T_i$, $U_L^L\myeq\sum_{\tau_i\in\tau_L} u_i^L/m$, $U_H^L\myeq\sum_{\tau_i\in\tau_H} u_i^L/m$ and $U_H^H\myeq\sum_{\tau_i\in\tau_H} u_i^H/m$. 

We assume that the tasks execute sequentially and are not allowed to simultaneously execute on more than one processor at any given time (i.e., $u_i^L\leq1$ and $u_i^H\leq1$). We consider the problem of scheduling this MC task system on a multiprocessor platform comprising $m$ cores.

\paragraph{MC Modes.} The system starts in \emph{LO-mode} and remains in that mode as long as all the tasks signal completion before exceeding LO-WCET values. The system switches to \emph{HI-mode} at the instant when any HI-task executes beyond its LO-WCET and does not signal completion.  \emph{Mode switch instant} is defined as the time instant when this mode change occurs. The system can safely return back to LO-mode at the time instant when all processors idle after the mode switch. We assume that no job of LO-task  $\tau_i$ would exceed $C_i^L$ and no job of HI-task $\tau_i$ would exceed $C_i^H$. 

\paragraph{MC-Schedulable.} A task set $\tau$ is said to be MC-schedulable by a scheduling algorithm if,
\begin{itemize}
\item LO-mode guarantee: Every job of each task in $\tau$ is able to complete LO-WCET execution within its deadline, and 
\item HI-mode guarantee: Every job of each task in $\tau_H$ is able to complete HI-WCET execution within its deadline.
\end{itemize}

Since no LO-task deadlines, including for jobs that are released before the mode switch, are required to be met in HI-mode, it is possible to drop all the LO-task jobs immediately upon mode switch. The HI-task jobs in HI-mode can be classified into two types depending on their release time. A job of task $\tau_i$ is said to be a \textbf{carry-over job}, if it is released before mode switch instant but has not completed its execution until mode switch. All the remaining HI-mode jobs of $\tau_i$, those that are released after the mode switch, are called \textbf{pure HI-mode jobs}. The carry-over and pure HI-mode jobs of all the HI-tasks must be guaranteed HI-WCET budgets by their deadlines. Similarly, the pure LO-mode jobs of all the tasks, those with deadlines before mode switch, must be guaranteed LO-WCET budgets by their deadlines. 

\subsection{Dual-rate Fluid Model}
\label{sec:dual_rate_model}
The dual-rate fluid model~\cite{lee_mcfluid} was designed to schedule implicit deadline task systems with two criticality levels. It assigns different execution rates to tasks in each criticality level.
The execution rate of a task is formally defined in Definition~\ref{execution_rate}.

\begin{definition}[Execution rate, from~\cite{lee_mcfluid}]
\label{execution_rate}
A task $\tau_i$ is said to be executed with execution rate $\theta_i \in \mathbb{R}^+$, s.t. $0 < \theta_i\leq 1$, if every job of the task is executed on a fractional processor with a speed of $\theta_i$.
\end{definition}
 
The dual-rate model can be summarized as follows: 

\begin{itemize}
\item Each task $\tau_i\in \tau$ executes at a constant rate $\theta_i^L$ (where $\theta_i^L \in [u_i^L,1]$) in the LO-mode,
\item All tasks in $\tau_L$ are discarded immediately upon mode switch, and 
\item Each task $\tau_i\in \tau_H$ executes at a constant rate $\theta_i^H$ (where $\theta_i^H \in [u_i^H,1]$) in the HI-mode.
\end{itemize}

An exact schedulability test for the dual-rate fluid model has also been derived (Theorem~1 in~\cite{lee_mcfluid}). Task set $\tau$ is said to be MC-schedulable under dual-rate fluid scheduling iff
\begin{equation}
\begin{split}
\forall \tau_i\in\tau,\hspace{10pt} \theta_i^L \geq u_i^L
\end{split}
\label{eqn:LO_task_feasibility}
\end{equation}
\begin{equation}
\begin{split}
\forall \tau_i\in\tau_H,\hspace{10pt} \frac{u_i^L}{\theta_i^L} + \frac{u_i^H-u_i^L}{\theta_i^H} \leq 1
\end{split}
\label{eqn:HI_task_feasibility1}
\end{equation}
\begin{equation}
\begin{split}
\forall \tau_i\in\tau_H,\hspace{10pt} \theta_i^H \geq \theta_i^L
\end{split}
\label{eqn:HI_task_feasibility2}
\end{equation}
\begin{equation}
\begin{split}
\sum_{\tau_i\in\tau} \theta_i^L \leq m
\end{split}
\label{eqn:LO_mode_feasibility}
\end{equation}
\begin{equation}
\begin{split}
\sum_{\tau_i\in\tau_H} \theta_i^H \leq m
\end{split}
\label{eqn:HI_mode_feasibility}
\end{equation}

Equations~\eqref{eqn:LO_mode_feasibility} and~\eqref{eqn:HI_mode_feasibility} ensure that the assigned rates are feasible in each mode on the multiprocessor platform (denoted as \textbf{platform feasibility tests}). Equation~\eqref{eqn:LO_task_feasibility} ensures that each task is schedulable in the LO-mode, i.e., each job of the task is able to receive sufficient budget (proportional to $u_i^L$) within its deadline (denoted as \textbf{LO-mode task schedulability test}). Likewise, Equations~\eqref{eqn:HI_task_feasibility1} and~\eqref{eqn:HI_task_feasibility2} ensure that each HI-task is schedulable in the HI-mode, including carry-over jobs (denoted as \textbf{HI-mode task schedulability test}). Note that although Equation~\eqref{eqn:HI_task_feasibility2} is not specified in Theorem~1 of~\cite{lee_mcfluid}, it is assumed in the derivation of the theorem. This HI-mode test is essentially derived by identifying a worst case mode switch instant for each HI-task. As shown in~\cite{lee_mcfluid}, the worst case mode switch occurs at the same time instant when a carry-over job of the task completes its LO-WCET, i.e., the task itself triggers the mode switch. This observation is intuitive, because in this case the carry-over job executes using the higher HI-mode rate for the shortest possible duration of time.

\section{Multi-rate Fluid Scheduling Model}
\label{sec:model}

Under the multi-rate fluid scheduling model that we propose in this paper, each task $\tau_i \in \tau$ executes with a single rate in the LO-mode as in the dual-rate model. But, unlike the dual-rate model, each task $\tau_i \in \tau_H$ executes with at most $n_H+1$ rates in the HI-mode, where $n_H = |\tau_H|$ is the number of HI-criticality tasks. The intuition behind these multiple rates in the HI-mode can be explained as follows. When a mode switch occurs, there are potentially $n_H$ carry over jobs in the system, all of which may require an execution rate higher than their HI-criticality utilization (i.e., $u_i^H$). By allowing them to adjust their execution rates several times in the \emph{transition period} immediately after a mode switch, it may be possible to accommodate more carry over demand in the system.

\paragraph{Multi-rate Fluid Scheduling Model}: The multi-rate fluid model can be formally defined as follows:

\begin{itemize}
\item \textit{LO-mode execution rate}: Each job of task $\tau_i \in \tau$ will start executing at a rate $\theta_i^L$ in the LO-mode. If there is no mode switch while the job is active, then this is the only rate at which it will execute.   
\item \textit{Transition execution rates}: Upon a mode switch, all the LO-tasks will be immediately dropped. The transition period after the mode switch is partitioned into $n_H$ transition windows ($j: 1 \leq j \leq n_H$), each of a fixed duration $w_j$. The jobs of a HI-task $\tau_i \in \tau_H$ execute at a constant rate of $\theta_{i,j}^{H}$ in each transition window $j$. 
\item \textit{HI-mode execution rate}: After the completion of the transition period (after $\sum_j w_j$ time units from the mode switch), each job of HI-task $\tau_i \in \tau_H$ executes at a constant rate of $\theta_i^H$. 
\end{itemize}

Thus, we denote the bounded period of time between the mode switch and $\sum_j w_j$ time units thereafter as the \emph{transition period}. Note that the execution rates in the transition period can be used either by carry over jobs or pure HI-mode jobs that are released in the transition period. We denote the HI-mode jobs released in the transition period as \textbf{transition jobs} and the HI-mode jobs released after the transition period as \textbf{stable jobs}. It is also worth noting that in this multi-rate model the $n_H$ transition window durations are all determined offline, and remain fixed at runtime. As a consequence, similar to the dual-rate model, the runtime scheduling mechanism for this multi-rate model is also very simple. 

The proposed multi-rate model is a generalization of the dual-rate model that was discussed in Sect.~\ref{sec:dual_rate_model}. If we set all the transition execution rates to be equal to $\theta_i^H$, then in the resulting model, each task executes with a fixed rate in LO-mode and a fixed rate in HI-mode. This setting is identical to the dual-rate model.
\section{Schedulability Test}
\label{sec:test}

In the previous section we presented the multi-rate model for dual-criticality implicit-deadline sporadic task system and the worst-case mode switch pattern for which the execution rates need to be determined. In this section we derive a sufficient schedulability test for the multi-rate fluid scheduling model defined in Sect.~\ref{sec:model}. We also show that this test is equivalent to the dual-rate schedulability test (Theorem~1 in~\cite{lee_mcfluid}) if we set all the transition rates to $\theta_i^H$. 
 
This derivation is comprised of four steps:
\begin{enumerate}
\item \textbf{LO-mode task schedulability test}: We derive a test using LO-mode execution rates to ensure that each task is schedulable in the LO-mode.
\item \textbf{LO-mode platform feasibility test}: We derive a test using LO-mode execution rates and the number of cores to ensure that the allocated LO-mode rates are feasible on the multiprocessor platform. 
\item \textbf{HI-mode platform feasibility test}: We derive a test using HI-mode execution rates and the number of cores to ensure that the allocated HI-mode rates are feasible on the multiprocessor platform. Here we need to consider all the $n_H+1$ rates assigned to each HI-task. 
\item \textbf{HI-mode task schedulability test}: We derive several conditions using LO-mode and HI-mode execution rates to ensure that each HI-task is schedulable in the HI-mode. We need to consider carry over jobs, transition jobs as well as stable jobs in this step. 
\end{enumerate}

\paragraph{LO-mode task schedulability test.} In the LO-mode, jobs of each task $\tau_i \in \tau$ execute with a single rate $\theta_i^L$. It is easy to see that $\theta_i^L \geq u_i^L$ is a necessary condition for task schedulability in LO-mode. Since our framework is based on fluid scheduling, this is also a sufficient schedulability condition in the LO-mode. If the task receives an execution rate of at least $u_i^L$, then the total allocated budget to each job of the task is at least $u_i^L \times T_i$, which is sufficient to meet deadlines. We record this test in the following proposition.

\begin{proposition} 
\label{prop:lo_sched}
Task system $\tau$ is schedulable in the LO-mode iff 
\begin{equation*}
\forall \tau_i \in \tau, \theta_i^L \geq u_i^L.
\end{equation*}
\end{proposition}

\paragraph{LO-mode platform feasibility test.} In the LO-mode, a set of execution rates is feasible on a multiprocessor platform comprising $m$ cores if and only if the sum total of the rates is no more than $m$. If the total is more than $m$, then clearly the rates cannot be assigned using $m$ cores. Whereas if it is no more than $m$, then it can be assigned because a single core can schedule multiple tasks at the same time under the fluid scheduling model.  

\begin{proposition}
\label{prop:lo_feas}
Execution rates assigned to task system $\tau$ are feasible in the LO-mode iff 
\begin{equation*}
\sum_{\tau_i\in\tau} \theta_i^L \leq m.
\end{equation*}
\end{proposition}

\paragraph{HI-mode platform feasibility test.} Similar to the LO-mode case, feasibility of the assigned execution rates is ensured as long as the sum total of the rates is no more than $m$. However, since we have multiple execution rates in the HI-mode, we need to ensure that rates are feasible in each of the $n_H+1$ windows.

\begin{proposition}
\label{prop:hi_feas}
Execution rates assigned to task system $\tau$ are feasible in the HI-mode iff 
\begin{align}
\forall j (1 \leq j \leq n_H), & \sum_{\tau_i\in\tau_H} \theta_{i,j}^H \leq m, ~\mbox{and} \label{eqn:transition_hi} \\
& \sum_{\tau_i\in\tau_H} \theta_i^H \leq m. \label{eqn:pure_hi}
\end{align}
\end{proposition}

\subsection{HI-mode task schedulability test} 

\begin{figure}
\centering
\includegraphics[width=0.75\textwidth]{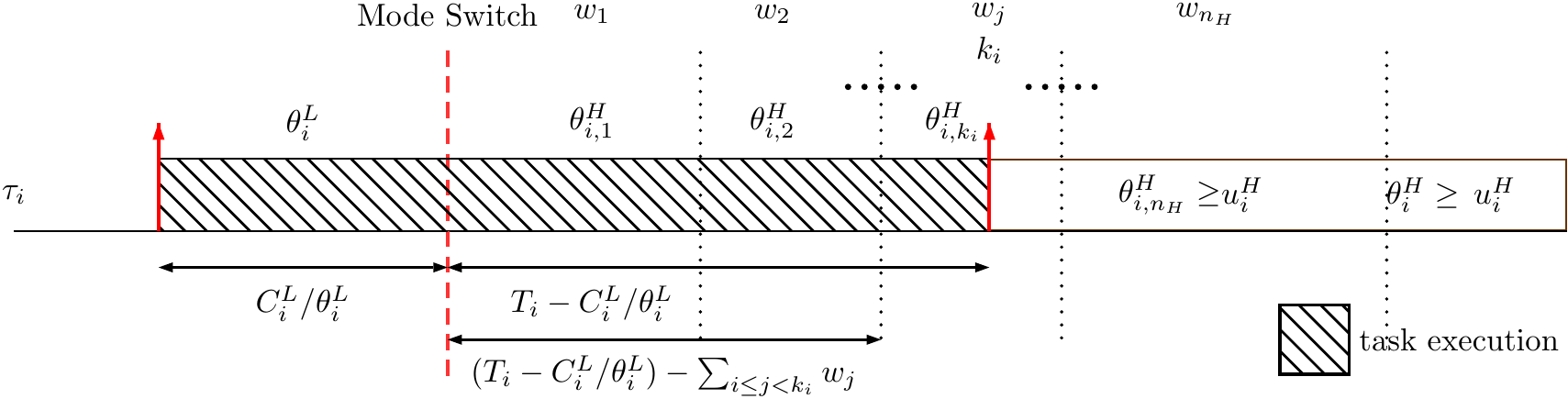}
\caption{Worst case mode switch instant of a carry over job}
\label{fig:carry_over_job_wc}
\end{figure}

In this section we derive sufficient schedulability conditions for tasks in the HI-mode. There are three types of jobs to consider. Carry over jobs that are released before the mode switch but remain unfinished in the LO-mode, transition jobs that are released in the transition period (within $\sum_{j: 1\leq j \leq n_H} w_j$ time units from the mode switch), and stable jobs that are released after the transition period. We first consider the stable jobs and then derive the tests for the other jobs. 

\paragraph{Stable jobs.} Since these jobs are released after the transition period, they always execute using a single rate $\theta_i^H$. As long as $\theta_i^H \geq u_i^H$ these jobs can meet their deadlines because $u_i^H * T_i = C_i^H$ units of execution would be guaranteed. It is easy to see that this is also a necessary condition for schedulability of the stable jobs; if $\theta_i^H < u_i^H$ then the stable jobs cannot meet their deadlines. We record this test in the following proposition.  

\begin{proposition}
\label{prop:stable_jobs}
A stable job of task $\tau_i$, that is a job released at or after $\sum_{j: 1\leq j \leq n_H} w_j$ time units from the mode switch, is schedulable iff 
\begin{equation*}
\theta_i^H  \geq u_i^H
\end{equation*}
\end{proposition} 

\paragraph{Carry over jobs.} Let $k_i$ denote the largest index ($1 \leq k_i \leq n_H+1$) such that the earliest deadline of any carry over job of task $\tau_i$ is strictly greater than $\sum_{j: 1\leq j < k_i} w_j$ time units after the mode switch (see Figure~\ref{fig:carry_over_job_wc}). That is, if a carry over job of $\tau_i$ executes for exactly $C_i^L$ time units and triggers the mode switch, then it has the earliest possible deadline in the HI-mode ($T_i - C_i^L/\theta_i^L$ time units from the mode switch instant). In this case, its deadline would either fall in the $k_i^{th}$ transition window if $k_i \leq n_H$ or after the transition period if $k_i = n_H+1$. Formally, $k_i$ can be defined as follows.

\begin{definition}[Earliest completion window $k_i$]
\label{def:k_i}
\begin{align*}
& \sum_{j: 1 \leq j < k_i} w_j < T_i - C_i^L/\theta_i^L, \\
& \mbox{and if $k_i \leq n_H$, then} \\
& \sum_{j: 1 \leq j \leq k_i} w_j \geq T_i - C_i^L/\theta_i^L.
\end{align*}
\end{definition}

Note that no job of task $\tau_i$ (carry over or otherwise) can have a deadline within $\sum_{j: 1 \leq j < k_i} w_j$ time units of the mode switch instant. This is an important property of $k_i$ that we will use in the following derivations. The below theorem derives a sufficient schedulability test for the carry over jobs. Similar to the dual-rate model, the main intuition behind this theorem is that the worst case scenario for a carry over job occurs when the job itself triggers the mode switch. It means that the job has executed with the smaller rate $\theta_i^L$ for the longest possible duration in the LO-mode and left with the shortest possible duration (i.e., $T_i - C_i^L/\theta_i^L$) in the HI-mode to complete its remaining execution (i.e., $C_i^H-C_i^L$).

\begin{figure}
\centering
\includegraphics[width=0.75\textwidth]{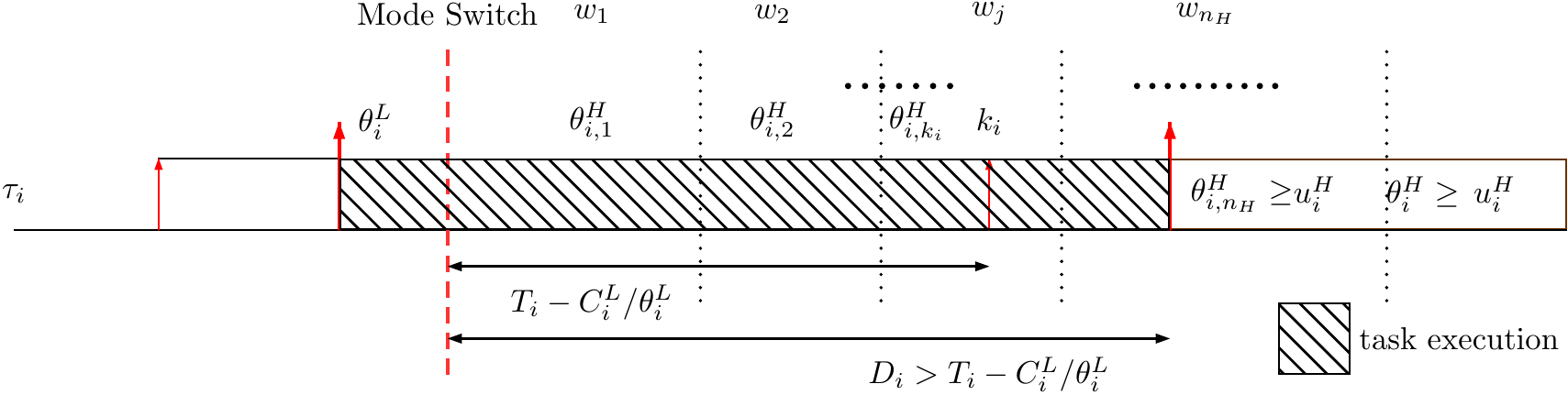}
\caption{A carry over job that does not trigger a mode switch}
\label{fig:carry_over_job}
\end{figure}

\begin{theorem}
\label{thm:co_hi_mode}
A carry over job of task $\tau_i \in \tau_H$ can meet its deadline in the HI-mode if the following three conditions are met.
\begin{align}
& \nonumber \sum_{j: 1 \leq j < k_i} \theta_{i,j}^H \times w_j + R_i \times \left (T_i - C_i^L/\theta_i^L - \sum_{j: 1 \leq j < k_i} w_j \right) \\
& \label{eqn:co_wc} \geq C_i^H - C_i^L \\
& \label{eqn:co_nwc1}
\forall j: k_i \leq j \leq n_H, \theta_i^L \leq \theta_{i,j}^H  \\
& \label{eqn:co_nwc2}
\theta_i^L \leq \theta_i^H  
\end{align}
where $R_i = \theta_{i,k_i}^H$ if $k_i \leq n_H$ and $R_i = \theta_i^H$ otherwise.
\end{theorem}
\begin{proof} 
\textbf{Only if:} We first show that Equation~\eqref{eqn:co_wc} is a necessary condition for the schedulability of a carry over job that triggers the mode switch. But since Equations~\eqref{eqn:co_nwc1} and~\eqref{eqn:co_nwc2} are not necessary, the theorem itself only presents a sufficient test. If Equation~\eqref{eqn:co_wc} does not hold then this carry over job will miss its deadline. Consider the carry over job shown in Figure~\ref{fig:carry_over_job_wc}. Since it triggers the mode switch, its remaining execution time at mode switch is $C_i^H - C_i^L$, and the time remaining to deadline is $T_i - C_i^L/\theta_i^L$. From the definition of $k_i$ we get $\sum_{j: 1 \leq j < k_i} w_j < T_i - C_i^L/\theta_i^L$, and the total execution for the carry over job in this window is $\sum_{j: 1 \leq j < k_i} \theta_{i,j}^H * w_j$. Thus the remaining execution of $C_i^H - C_i^L - \sum_{j: 1 \leq j < k_i} \theta_{i,j}^H * w_j$ must be provided in the remaining time window of $T_i - C_i^L/\theta_i^L - \sum_{j: 1 \leq j < k_i} w_j$. This is only possible if Equation~\eqref{eqn:co_wc} holds.

\textbf{If:} Assuming Equation~\eqref{eqn:co_wc} holds, the total execution a carry over job triggering a mode switch receives after the mode switch is $\geq C_i^H-C_i^L$. Since it has already received $C_i^L$ units of execution before mode switch, it can meet its deadline. Let us consider the case of a carry over job that does not trigger the mode switch as shown in Figure~\ref{fig:carry_over_job}. In this case the carry over job has a deadline greater than $T_i - C_i^L/\theta_i^L (> \sum_{j: 1 \leq j < k_i} w_j)$. Suppose its deadline is at some time instant $t$. Then, from Equations~\eqref{eqn:co_nwc1} and~\eqref{eqn:co_nwc2} we get that each execution rate assigned to this carry over job in the time interval $[T_i - C_i^L/\theta_i^L, t]$ is at least as much as $\theta_i^L$. Hence, the amount of execution over and above $C_i^H-C_i^L$ remaining at mode switch ($(t-T_i - C_i^L/\theta_i^L) \times \theta_i^L$) is guaranteed to be provided by the job's deadline in the HI-mode.\qed
\end{proof}

\subsubsection{Schedulability tests for transition HI-mode jobs} 

Transition HI-mode jobs are released in the transition period (within $\sum_{j: 1\leq j \leq n_H} \allowbreak w_j$ time units from the mode switch). We further classify them into two categories depending on whether they are released before the $k_i^{th}$ window or not. Transition jobs released before $\sum_{j: 1\leq j < k_i} w_j$ are denoted as \textbf{early transition jobs}, and the remaining transition jobs are denoted as \textbf{late transition jobs}. Note that by definition, when $k_i = n_H+1$ all the transition jobs are in fact early transition jobs. We now derive schedulability conditions for each of these jobs separately. Figure~\ref{fig:transition_jobs} shows examples of both these transition jobs.
\begin{figure}
	    \centering
	    \subfigure[Early transition job]{
            	\includegraphics[width=0.75\textwidth]{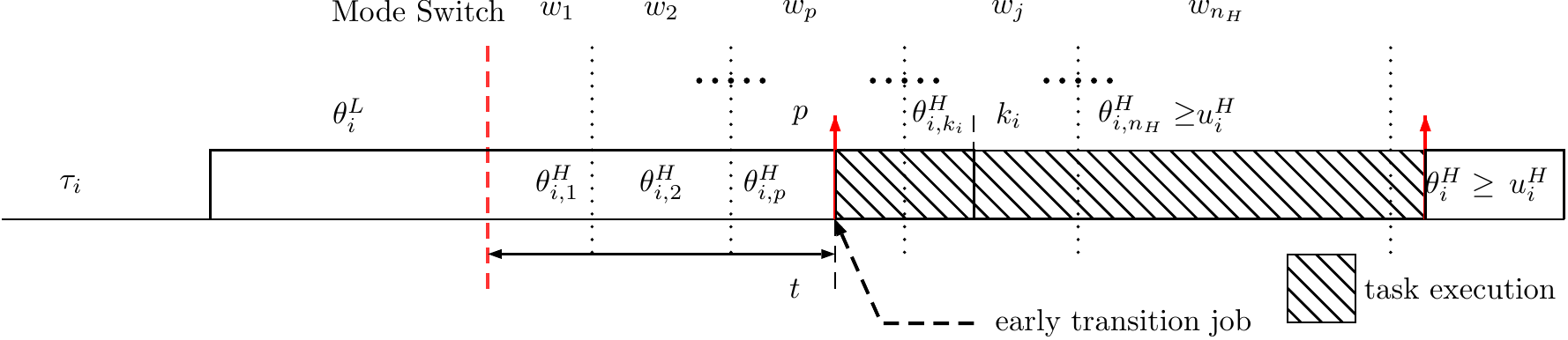}
            	\label{fig:early_transition_job}
		}
	     \subfigure[Late transition job]{
            	\includegraphics[width=0.75\textwidth]{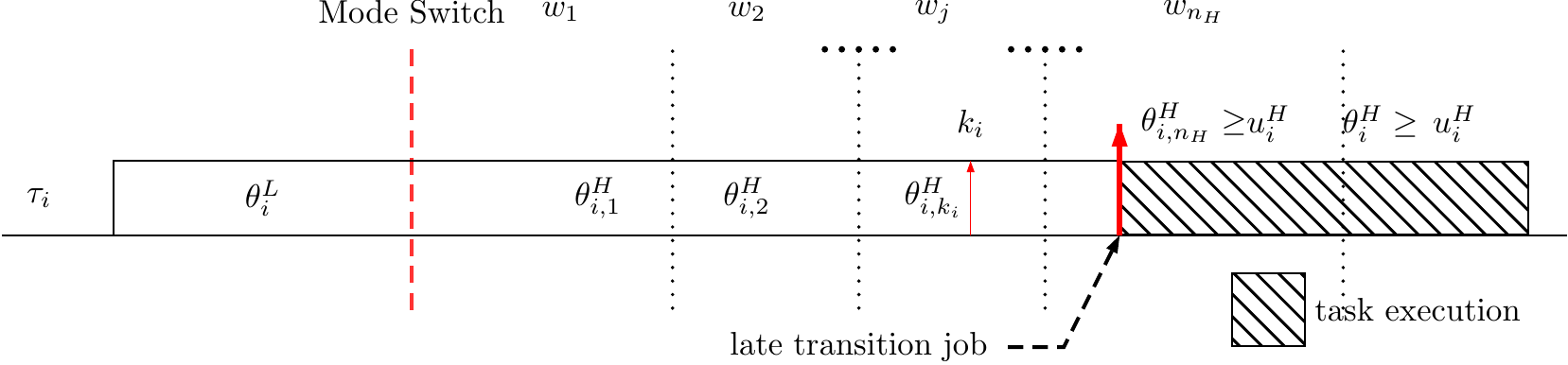}
            	\label{fig:late_transition_job}
	     }
\caption{Transition jobs}
\label{fig:transition_jobs}
\end{figure}

\paragraph{Early transition jobs.} Early transition jobs are released in the transition period and use the execution rates in the transition windows before $k_i$. To derive the schedulability conditions for this case, we first present a simple property of non-decreasing execution rates.

\begin{lemma}
\label{lem:1}
Suppose the execution rates in the transition windows prior to $k_i$ satisfy the following conditions.
{\small \begin{align}
& \label{eqn:lem_transition1}
\sum_{j: 1\leq j < k_i} \theta_{i,j}^H \times w_j \geq  u_i^H \times \sum_{j: 1\leq j < k_i} w_j  \\
& \label{eqn:lem_transition2}
\forall j: 1\leq j < k_i, \theta_{i,j}^H \leq \theta_{i,j+1}^H 
\end{align}}

Then, in any contiguous time period after the mode switch but immediately preceding the $k_i^{th}$ window (say $[t, \sum_{j: 1 \leq j < k_i} w_j]$ where $t$ is a time instant relative to the mode switch instant), the total execution assigned to $\tau_i$ is at least $(\sum_{j: 1 \leq j < k_i} w_j - t)u_i^H$. 
\end{lemma}
\begin{proof}
Consider the time interval $[t, \sum_{j: 1 \leq j < k_i} w_j]$, and let $\sum_{j: 1\leq j < p} w_j \leq t < \sum_{j: 1\leq j \leq p} w_j$ for some transition window $p < k_i$. $p$ denotes the transition window to which time instant $t$ belongs (see Figure~\ref{fig:early_transition_job}). We consider two cases depending on whether $\theta_{i,p}^H \geq u_i^H$ or not.

\emph{Case 1 ($\theta_{i,p}^H \geq u _i^H$):} In this case, using Equation~\eqref{eqn:lem_transition2} we get that each execution rate allocated in the time interval $[t, \sum_{j: 1 \leq j < k_i} w_j]$ is at least $u_i^H$. Therefore, it is true that the total execution assigned to $\tau_i$ is at least $(\sum_{j: 1 \leq j < k_i} w_j - t)u_i^H$.  

\emph{Case 2 ($\theta_{i,p}^H < u _i^H$):} Consider Equation~\eqref{eqn:lem_transition1}.
\begin{align*}
\sum_{j: 1\leq j < k_i} \theta_{i,j}^H \times w_j & \geq  u_i^H \times \sum_{j: 1\leq j < k_i} w_j
\end{align*}
\begin{align*}
&\mbox{(Using $\forall j: j < p, \theta_{i,p}^H \geq \theta_{i,j}^H$ from Equation~\eqref{eqn:lem_transition2})}\hspace{1cm}\\
\Rightarrow & \sum_{j: 1\leq j \leq p} \theta_{i,p}^H \times w_j + \sum_{j: p < j < k_i} \theta_{i,j}^H \times w_j \geq  \left(u_i^H \times \sum_{j: 1\leq j < k_i} w_j\right)\\
&\mbox{(Adding and subtracting $t \times \theta_{i,p}^H$)}\hspace{1cm}\\
\Rightarrow & t \times \theta_{i,p}^H + \left(\sum_{j: 1\leq j \leq p} w_j - t\right)\times \theta_{i,p}^H + \sum_{j: p < j < k_i} \theta_{i,j}^H \times w_j \geq  \left(u_i^H \times \sum_{j: 1\leq j < k_i} w_j\right) \\
\Rightarrow & \left(\sum_{j: 1\leq j \leq p} w_j - t\right) \times \theta_{i,p}^H + \sum_{j: p < j < k_i} \theta_{i,j}^H \times w_j
\geq  \left(u_i^H \times \sum_{j: 1\leq j < k_i} w_j\right)- t \times \theta_{i,p}^H\\
\Rightarrow & \left(\sum_{j: 1\leq j \leq p} w_j - t\right) \times \theta_{i,p}^H + \sum_{j: p < j < k_i} \theta_{i,j}^H \times w_j
\geq  u_i^H \times \left(\sum_{j: 1\leq j < k_i} w_j - t\right)\\
&\hspace{9cm} \mbox{(Using $\theta_{i,p}^H < u_i^H$)}
\end{align*}  
The left hand side in the above inequality denotes the exact total execution allocated to $\tau_i$ in the time interval $[t, \sum_{j: 1 \leq j < k_i} w_j]$. Therefore, the lemma holds true for this case as well.\qed
\end{proof}

Condition~\eqref{eqn:lem_transition2} is necessary because the tasks that are released immediately after mode switch (within the transition window) will not be schedulable otherwise. Although it seems to be a restriction on the model, it is in fact reasonable because one can expect a higher rate assignment as more carry-over jobs complete.

The following theorem presents a sufficient schedulability condition for the early transition jobs using the above lemma.

\begin{theorem}
\label{thm:transition_hi_mode}
An early transition job of task $\tau_i \in \tau_H$ can meet its deadline in the HI-mode if the following four conditions are met.
\begin{align}
& \label{eqn:transition1}
\sum_{j: 1\leq j < k_i} \theta_{i,j}^H \times w_j \geq  u_i^H \times \sum_{j: 1\leq j < k_i} w_j  \\
& \label{eqn:transition2}
\forall j: 1\leq j < k_i, \theta_{i,j}^H \leq \theta_{i,j+1}^H  \\
& \label{eqn:transition3}
\forall j: k_i \leq j \leq n_H, \theta_{i,j}^H \geq u_i^H  \\
& \label{eqn:transition4}
\theta_i^H  \geq u_i^H
\end{align}
\end{theorem}
\begin{proof}
Consider an early transition job released at some time instant $t$ after the mode switch as shown in Figure~\ref{fig:early_transition_job}. By definition, $t < \sum_{j: 1 \leq j < k_i} w_j, t < T_i - C_i/\theta_i^L$ and the job's deadline is no earlier than $T_i - C_i/\theta_i^L$. The latter follows from the fact that the job is released after the mode switch, and hence its deadline cannot be earlier than the earliest possible deadline of a carry over job.

Consider the time interval $[t, \sum_{j: 1 \leq j < k_i} w_j]$. Using Equations~\eqref{eqn:transition1} and~\eqref{eqn:transition2} and Lemma~\ref{lem:1} we get that the total execution assigned to the job in this interval is at least $\big(\sum_{j: 1 \leq j < k_i} w_j - t\big)\times u_i^H$. In each succeeding window, the job is assigned an execution rate of at least $u_i^H$ (Equations~\eqref{eqn:transition3} and~\eqref{eqn:transition4}). 
Therefore, the total execution assigned to the job between its release time and deadline is at least $C_i^H$, and hence the job can meet its deadline.\qed
\end{proof}

\paragraph{Late transition jobs.} The only remaining jobs to consider are late transition jobs that are released in the transition period but no earlier than $\sum_{j: 1 \leq j < k_i} w_j$. Since these jobs only use the execution rates in windows at or after $k_i$, Equations~\eqref{eqn:transition3} and~\eqref{eqn:transition4} in Theorem~\ref{thm:transition_hi_mode} are sufficient to guarantee the schedulability of these jobs. If each of those rates are at least $u_i^H$, then the late transition jobs are guaranteed to execute for $C_i^H$ time units by their deadlines.

Thus, combining Propositions~\ref{prop:lo_sched}--\ref{prop:hi_feas} and Theorems~\ref{thm:co_hi_mode} and~\ref{thm:transition_hi_mode} we obtain a sufficient schedulability test for the multi-rate fluid scheduling model. Also, note that Proposition~\ref{prop:stable_jobs} is subsumed by Theorem~\ref{thm:transition_hi_mode}. Note that the test has polynomial time complexity because all the equations can be evaluated in time proportional to the number of tasks in $\tau$.

\begin{example} Consider the task set $\tau$ with multi-rate assignments as shown in Table~\ref{tab:example}. We show that this rate and window assignment is schedulable. We can easily check that Propositions~\ref{prop:lo_sched}--\ref{prop:hi_feas}, Equations~\eqref{eqn:co_nwc1} and~\eqref{eqn:co_nwc2} of Theorem~\ref{thm:co_hi_mode} and all the equations of Theorem~\ref{thm:transition_hi_mode} are satisfied. Now consider Equation~\eqref{eqn:co_wc} of Theorem~\ref{thm:co_hi_mode}: for $\tau_1$, $1 \times 2.1 \geq 4.9-2.8$; for $\tau_2$, $1 \times 2.1 + 1 \times 0.4 \geq 4-1.5$; for $\tau_3$, $0 \times 2.1 + 0.3 \times 0.4 + 0.5 \times 13.76 \geq 10.5-3.5$. Thus, Equation~\eqref{eqn:co_wc} is also satisfied, and hence the task set is schedulable.
\end{example}
\section{Properties of the multi-rate fluid model}
\label{sec:property}

In this section we derive some important properties of the proposed multi-rate fluid scheduling model and its schedulability test. A \emph{rate and window assignment algorithm} for the multi-rate fluid scheduling model is an algorithm that assigns values for all the execution rates in the model as well as for the $n_H$ transition window durations. We define the correctness criteria for such an algorithm as follows.  

\begin{definition}[Multi-Rate MC-Correctness]
\label{Multi-Rate MC-Correctness}
A rate and window assignment algorithm is called \textbf{Multi-Rate MC Correct} iff the following holds: Whenever the algorithm returns a set of execution rates ($\forall i: \theta_i^L, \theta_i^H$ and $\forall i,j: \theta_{i,j}^H$) and transition window durations ($\forall j: w_j$), these rates and window durations satisfy Propositions~\ref{prop:lo_sched}--\ref{prop:hi_feas} and Theorems~\ref{thm:co_hi_mode} and~\ref{thm:transition_hi_mode}. 
\end{definition}

\paragraph{Dual-rate generalization and speed-up optimality.} An interesting property of the proposed multi-rate model and schedulability test is that it generalizes the dual-rate model and schedulability test discussed in Sect.~\ref{sec:system_model}.
We can obtain the dual-rate model by setting $\forall j: \theta_{i,j}^H = \theta_i^H$ for each HI-task $\tau_i$. 
For this case, the following lemma shows that the multi-rate schedulability test derived in Sect.~\ref{sec:test} is equivalent to the dual-rate schedulability test.

\begin{lemma}[Dual-rate Generalization]
If $\forall i: 1\leq i \leq n_H, \forall j: 1\leq j \leq n_H, \theta_{i,j}^H = \theta_i^H$ in the multi-rate fluid model, then Propositions~\ref{prop:lo_sched}--\ref{prop:hi_feas} and Theorems~\ref{thm:co_hi_mode} and~\ref{thm:transition_hi_mode} are satisfied if and only if Equations~\eqref{eqn:LO_task_feasibility}--\eqref{eqn:HI_mode_feasibility} hold.
\label{lem:generalization}
\end{lemma}
\begin{proof}
The conditions for LO-mode platform feasibility and LO-task schedulability are identical for the dual-rate and multi-rate models (Proposition~\ref{prop:lo_sched} and Equation~\eqref{eqn:LO_task_feasibility}, Proposition~\ref{prop:lo_feas} and Equation~\eqref{eqn:LO_mode_feasibility}).

We now consider the HI-task schedulability and HI-mode platform feasibility tests. 

\paragraph{(Dual-rate $\Rightarrow$ Multi-rate):} Suppose the tests are satisfied for the dual-rate model. That is, Equations~\eqref{eqn:HI_task_feasibility1},~\eqref{eqn:HI_task_feasibility2} and~\eqref{eqn:HI_mode_feasibility} hold.

By substituting Equation~\eqref{eqn:HI_task_feasibility2} in Equation~\eqref{eqn:HI_task_feasibility1} we get,
\begin{align}
\nonumber & 1 \geq \frac{u_i^H - u_i^L}{\theta_i^H} + \frac{u_i^L}{\theta_i^L} & \mbox{Equation~\eqref{eqn:HI_task_feasibility1}} \\
\nonumber \Leftrightarrow & 1 \geq \frac{u_i^H - u_i^L}{\theta_i^H} + \frac{u_i^L}{\theta_i^H} & \mbox{(Using $\theta_i^L \leq \theta_i^H$)} \\
\Leftrightarrow & \theta_i^H \geq u_i^H \label{eqn:1} &
\end{align}

First we consider Theorem~\ref{thm:transition_hi_mode}. Equation~\eqref{eqn:transition4} in this theorem is identical to Equation~\eqref{eqn:1} above. Equations~\eqref{eqn:transition1},~\eqref{eqn:transition2} and~\eqref{eqn:transition3} are all satisfied because $\theta_{i,j}^H = \theta_i^H \geq u_i^H$ for all $i$ and $j$ based on the assumption of the lemma and Equation~\eqref{eqn:1}.

Next we consider Proposition~\ref{prop:hi_feas}. Equation~\eqref{eqn:pure_hi} in this proposition is identical to Equation~\eqref{eqn:HI_mode_feasibility}. Also, Equation~\eqref{eqn:transition_hi} reduces to $\sum_{\tau_i\in\tau_H} \theta_i^H \leq m$, which is also satisfied based on Equation~\eqref{eqn:HI_mode_feasibility}. Thus the proposition is satisfied.

Finally, we consider Theorem~\ref{thm:co_hi_mode}. Equation~\eqref{eqn:co_wc} in Theorem~\ref{thm:co_hi_mode} is 

 \begin{align*}
& \sum_{j: 1 \leq j < k_i} \theta_{i,j}^H \times w_j + \theta_i^H \times \left (T_i - C_i^L/\theta_i^L - \sum_{j: 1 \leq j < k_i} w_j \right) \\
& \geq C_i^H - C_i^L
\end{align*}
\begin{align*}
\Leftrightarrow & \theta_i^H(T_i - C_i^L/\theta_i^L) \geq C_i^H - C_i^L &\mbox{(Using $\theta_{i,j}^H = \theta_i^H, \forall i,j$)} \\
\Leftrightarrow & \theta_i^HT_i(1 - u_i^L/\theta_i^L) \geq T_i(u_i^H - u_i^L) \\
\Leftrightarrow & 1- u_i^L/\theta_i^L \geq \frac{u_i^H - u_i^L}{\theta_i^H}\\
\Leftrightarrow & 1 \geq \frac{u_i^H - u_i^L}{\theta_i^H} + \frac{u_i^L}{\theta_i^L} &\mbox{(Equation~\eqref{eqn:HI_task_feasibility1})}
\end{align*}

Equations~\eqref{eqn:co_nwc1} and~\eqref{eqn:co_nwc2} reduce to $\theta_i^L \leq \theta_i^H$ using $\theta_{i,j}^H = \theta_i^H, \forall i,j$. This is identical to Equation~\eqref{eqn:HI_task_feasibility2}. Thus, we have shown that satisfaction of the dual-rate schedulability test also implies satisfaction of the multi-rate schedulability test. 

\paragraph{(Multi-rate $\Rightarrow$ Dual-rate):} This case is trivial because Proposition~\ref{prop:hi_feas}, and Theorems~\ref{thm:co_hi_mode} and~\ref{thm:transition_hi_mode} subsume Equations~\eqref{eqn:HI_task_feasibility1},~\eqref{eqn:HI_task_feasibility2} and~\eqref{eqn:HI_mode_feasibility}.\qed 
\end{proof}

Note that, as discussed in the introduction, dual-rate assignment strategy MC-Fluid has been shown to be speed-up optimal with a speed-up bound of $4/3$~\cite{baruah_mcf}. The above generalization lemma is then important because it can be used to show that rate and window assignment algorithms for the multi-rate model satisfying some properties are also speed-up optimal. In particular, we will show that rate and window assignment algorithms for the multi-rate model that dominate MC-Fluid in terms of schedulability are also speed-up optimal with a speed-up bound of $4/3$~\footnote{Similar proof-technique has been used in~\cite{baruah_mcf}, in which dominance of MC-Fluid over MCF has been used to derive a speed-up bound for MC-Fluid.}.

Since MC-Fluid has a speed-up bound of $4/3$, this means any MC task system satisfying the condition $\max \{ U_L^L+U_H^L, U_H^H \} \leq 3/4$ and $\max_i \{ u_i^L, u_i^H \} \leq 3/4$ is schedulable under MC-Fluid. Then, if some multi-rate and window assignment algorithm A dominates MC-Fluid, it implies that any MC task system satisfying these conditions is also schedulable under A. This directly implies a speed-up bound of $4/3$ for A, because any feasible MC task system will satisfy the above conditions when the speed of each processor is increased by a factor of $4/3$. Based on this intuition, the following lemma derives properties for A that ensures A dominates MC-Fluid.     

\begin{lemma}[Dominance and Speed-up Bound]
\label{lemma:speed-up}
Suppose a rate and window assignment algorithm A for the multi-rate fluid model is guaranteed to return some feasible assignment as long as there is at least one assignment satisfying,
\begin{itemize}
\item $\forall i: 1\leq i \leq n_H, \forall j: 1\leq j \leq n_H, \theta_{i,j}^H = \theta_i^H$, 
\item $\forall j: 1 \leq j \leq n_H, w_j = 0$, and
\item Propositions~\ref{prop:lo_sched}--\ref{prop:hi_feas} and Theorems~\ref{thm:co_hi_mode} and~\ref{thm:transition_hi_mode}.
\end{itemize}
Then Algorithm A dominates MC-Fluid and has a speed-up bound of $4/3$.
\end{lemma}
\begin{proof}
Suppose a MC task system is feasible under MC-Fluid. Then, let $\theta_i^L$ and $\theta_i^H$ denote the execution rates for each task $\tau_i \in \tau_H$ assigned by MC-Fluid that satisfies Equations~\eqref{eqn:LO_task_feasibility}--\eqref{eqn:HI_mode_feasibility}. Now consider the multi-rate assignment $\forall i: 1\leq i \leq n_H, \forall j: 1\leq j \leq n_H, \theta_{i,j}^H = \theta_i^H$, $\forall j: 1 \leq j \leq n_H, w_j = 0$ and $\theta_i^L$ identical to the dual-rate assignment. From Lemma~\ref{lem:generalization} we know that this rate assignment satisfies Propositions~\ref{prop:lo_sched}--\ref{prop:hi_feas} and Theorems~\ref{thm:co_hi_mode} and~\ref{thm:transition_hi_mode}. Thus, there is at least one rate and window assignment that satisfies all the above three conditions of the lemma. From the assumption of the lemma, we then know that algorithm A is guaranteed to return a feasible multi-rate and window assignment. This shows that algorithm A dominates MC-Fluid in terms of schedulability. Since MC-Fluid has a known speed-up bound of $4/3$, by definition algorithm A also has a speed-up bound of $4/3$. \qed
\end{proof}

It has been shown that no non-clairvoyant scheduling algorithm can have a speed-up bound lower than $4/3$ for scheduling dual-criticality implicit-deadline sporadic task systems on multiprocessor platforms~\cite{baruah_edfvd2}. The speed-up optimality of the multi-rate and window assignment algorithm A with a speed-up bound of $4/3$ follows by combining Lemma~\ref{lemma:speed-up} and the fact that no non-clairvoyant algorithm can have a speed-up bound smaller than $4/3$.

\paragraph{Mapping to non-fluid platform.} The multi-rate fluid scheduling model, similar to other fluid models, assumes that a processing core can be fractionally assigned to tasks. Since this is not possible on a real (non-fluid) platform, it is important to map the fluid execution rates to a non-fluid scheduling policy, ideally without any loss in schedulability. In~\cite{lee_mcfluid}, the dual-rate fluid model has been successfully mapped without any loss in schedulability to the DP-Fair non-fluid scheduling algorithm (short for Deadline-Partitions Fair)~\cite{funk_dpfair}.
We now show that the multi-rate fluid model can also be similarly mapped to DP-Fair without any loss in schedulability.
In the classic non-MC multiprocessor scheduling, DP-Fair scheduling algorithm has been used to map a fluid execution model with single execution rate for each task to a non-fluid schedule without any loss in schedulability~\cite{funk_dpfair}. The main intuition behind this mapping is as follows. The entire scheduling window is partitioned based on job deadlines. Between any two consecutive job deadlines, it is ensured that the total allocation to a task is equal to the length of the partition multiplied by its assigned execution rate. Thus, \emph{tasks are guaranteed processor allocations proportional to their execution rates at every job deadline}, and hence schedulability is preserved. Between two consecutive job deadlines, the processor share allocated to each task can be scheduled using any non-fluid optimal scheduling policy such as McNaughton's algorithm~\cite{mcnaughton}. To handle sporadic job releases, the allocations in each partition are scheduled using a work-conserving algorithm, and they are re-computed upon job releases. Note that this scheduling policy can be implemented fully online by computing allocations for pending jobs between the current time instant (coinciding with a job release or deadline) and the next earliest job deadline.

In the multi-rate fluid model, apart from job deadlines, another important event is the mode switch instant. A mode switch is triggered when a HI-job executes for its LO-WCET ($C_i^L$) and does not signal completion. For correctness of the multi-rate schedulability test derived in Sect.~\ref{sec:test}, it is essential to ensure that each HI-job of a task $\tau_i \in \tau_H$ completes its LO-WCET execution $C_i^L$ no later than $C_i^L/\theta_i^L$ time units from its release. This is necessary to guarantee the property that any job deadline of task $\tau_i$ in HI-mode is no earlier than $T_i - C_i/\theta_i^L$ time units from the mode switch, which is used in Theorems~\ref{thm:co_hi_mode} and~\ref{thm:transition_hi_mode}. Hence, when mapping the multi-rate fluid model to DP-Fair scheduling, in addition to ensuring task allocation fairness at job deadlines, we also need to ensure this fairness at \emph{worst case mode switch instants} for each HI-job. That is, we need to ensure fairness at $C_i^L/\theta_i^L$ time units from each job release instant for each HI-task $\tau_i$.

After the mode switch, one simple way to map the multi-rate model to DP-Fair scheduling is to ensure allocation fairness at the boundary of each transition window $j$ ($1 \leq j \leq n_H$), in addition to job deadlines. Of course, this is not necessary, and schedulability would be guaranteed even if fairness is only ensured at job deadlines. Suppose a partition (interval between two consecutive job deadlines or between the mode switch and the earliest job deadline) spans more than one transition window. Let $[t_1, t_2]$ denote this partition, $j_1 (\leq n_H)$ denote a transition window that contains $t_1$, and $j_2 (\leq n_H)$ denote another transition window that contains $t_2$. Then, for each HI-task $\tau_i$ it is sufficient to ensure that its total allocation in this partition is equal to $(\sum_{1 \leq j \leq j_1} w_j - t_1)\theta_{i,j_1}^H + (\sum_{j_1 < j < j_2} w_j)\theta_{i,j}^H + (t_2 - \sum_{1 \leq j < j_2} w_j)\theta_{i,j_2}^H$. 

Thus, the multi-rate MC fluid scheduling model can be mapped to the DP-Fair scheduling algorithm using the following strategy.

\begin{definition}[DP-Fair Mapping]
\label{def:dp_fair}
Given a rate and window assignment in the multi-rate fluid model, consider the following task allocations.
\begin{itemize}
\item \textbf{LO-mode mapping:} Partition the scheduling window based on job deadlines and worst case mode switch instants ($C_i^L/\theta_i^L$ time units from job release for each HI-task $\tau_i$). In each partition $p$ having length $L_p$, allocate an execution of $\theta_i^L \times L_p$ time units for each task $\tau_i$ that has an unfinished job at the beginning of the partition.  
\item \textbf{Transition period mapping:} After mode switch drop all the jobs of LO-tasks. Partition the transition period (within $\sum_{1\leq j \leq n_H} w_j$ time units from the mode switch) based on HI-job deadlines. Consider a partition $p$ spanning the interval $[t_1, t_2]$, such that $j_1 (\leq n_H)$ denotes a transition window that contains $t_1$, and $j_2 (\leq n_H)$ denotes another transition window that contains $t_2$. Allocate an execution of $(\sum_{1 \leq j \leq j_1} w_j - t_1)\theta_{i,j_1}^H + (\sum_{j_1 < j < j_2} w_j)\theta_{i,j}^H + (t_2 - \sum_{1 \leq j < j_2} w_j)\theta_{i,j_2}^H$ time units for each HI-task $\tau_i$ that has an unfinished job at $t_1$.
\item \textbf{HI-mode mapping:} After $\sum_{1\leq j \leq n_H} w_j$ time units from the mode switch, partition the scheduling window based on HI-job deadlines. In each partition $p$ having length $L_p$, allocate an execution of $\theta_i^H \times L_p$ time units for each HI-task $\tau_i$ that has an unfinished job at the beginning of the partition.
\end{itemize}
In each partition, schedule the task allocations using any work-conserving optimal algorithm.
\end{definition}

The following lemma records the fact that the above mapping is optimal, in the sense that there is no loss in schedulability.

\begin{lemma}
Consider a set of execution rates and window durations that satisfy Propositions~\ref{prop:lo_sched}--\ref{prop:hi_feas} and Theorems~\ref{thm:co_hi_mode} and~\ref{thm:transition_hi_mode}. If these rates are mapped to the DP-Fair scheduling algorithm using Definition~\ref{def:dp_fair}, then in the resulting schedule all job deadlines are met.  
\end{lemma}
\begin{proof}
In the LO-mode, each job of a task $\tau_i$ is guaranteed to receive execution proportional to its assigned rate $\theta_i^L$ within its deadline. 
This follows from Proposition~\ref{prop:lo_feas} and the fact that allocations in each partition are scheduled using a work-conserving optimal algorithm. Then, using Proposition~\ref{prop:lo_sched} we can conclude that all the job deadlines in the LO-mode are met. 

Each job of a HI-task $\tau_i$ is guaranteed to receive at least $C_i^L$ time units of execution within $C_i^L/\theta_i^L$ time units from its release. Therefore, the deadline of any job of $\tau_i$ in the HI-mode is no earlier than $T_i-C_i^L/\theta_i^L$ time units from the mode switch instant. 

In the HI-mode, each job is guaranteed to receive execution proportional to the assigned rates within its deadline (using Proposition~\ref{prop:hi_feas} and the fact that allocations are scheduled using a work-conserving optimal algorithm). Therefore, using Theorems~\ref{thm:co_hi_mode} and~\ref{thm:transition_hi_mode} we can conclude that all job deadlines in the HI-mode will be met.\qed
\end{proof}
\section{Multi-rate Assignment Strategy}
\label{sec:rates}
In this section, we present an algorithm called \emph{SOMA} (short for Speed-up Optimal Multi-rate Assignment) to determine the execution rates and window durations for the multi-rate model.  A convex optimization based solution, similar to the dual-rate algorithm MC-Fluid, is desirable because it can optimize the assignments. However, one major hurdle towards using an optimization framework for the multi-rate model is that the earliest completion window parameters, $k_i$s in Definition~\ref{def:k_i}, control the number of terms involved in the summations used in the schedulability test. Since $k_i$s are dependent on the execution rates and window durations, this would imply that the constraints of the optimization problem are no longer fixed. To address this issue, we first use a simple heuristic to fix these $k_i$s and then formulate the optimization problem.

Algorithm~\ref{algo:1} presents SOMA. Since it is sufficient for all jobs of LO-tasks to execute at their minimum required rate, we assign $\theta_i^L = u_i^L$ for all $\tau_i  \in \tau_L$. The algorithm first sorts all the HI-tasks in increasing order of the parameter $T_i - C_i^L/u_i^H$. Note that $u_i^H$ denotes the maximum possible rate that a HI-task $\tau_i$ can be assigned in the LO-mode (easily seen by combining equations in Theorems~\ref{thm:co_hi_mode} and~\ref{thm:transition_hi_mode}). Then $T_i - C_i^L/u_i^H$ denotes the earliest completion time for any carry over job of this task, assuming the maximum possible rate in the LO-mode. Hence, in any feasible assignment, a HI-task $\tau_i$ with a smaller $T_i - C_i^L/u_i^H$, will most likely end up with a smaller $k_i$ when compared to a HI-task $\tau_j$ with a larger $T_j - C_j^L/u_j^H$. We use this intuition to fix the values for $k_i$s, and thus address the challenge discussed above. Inspite of this restriction, we show that SOMA has a speed-up bound of $4/3$ (Lemma~\ref{cor:soma_speed_up} below) and has comparable performance to a brute-force technique that tries all the $n_H!$ possible combinations for the $k_i$s (Sect.~\ref{sec:experiments}).

\begin{algorithm}
\begin{algorithmic}[1]
\Require $\tau$, $m$
\State For each $\tau_i \in \tau_L$ assign $\theta_i^L=u_i^L$.
\State Sort $\tau_H$ in increasing order of the parameter $T_i - (C_i^L/u_i^H)$. 
\State Solve the optimization problem in Definition~\ref{def:mr_opt}, assuming task indices are sorted based on the above sorting order.
\If{Optimization returns an assignment}
	\If{$\sum_{\tau_i \in \tau} \theta_i^L \leq m$}
		\State Declare Success
	\Else
		\State Declare Failure
	\EndIf
\Else
	\State Declare Failure
\EndIf
\end{algorithmic}
\caption{\label{algo:1}SOMA rate and window assignment strategy}
\end{algorithm}

SOMA then solves the convex optimization problem given in Definition~\ref{def:mr_opt}. Using a fixed value for all the $k_i$s ($k_i = i$), where we assume task indices are sorted based on increasing value of the parameter $T_i - (C_i^L/u_i^H)$, this optimization encodes Propositions~\ref{prop:lo_sched} and~\ref{prop:hi_feas} and Theorems~\ref{thm:co_hi_mode} and~\ref{thm:transition_hi_mode} of the schedulability test in its constraints. Its objective is to minimize the total LO-mode execution rates. If it returns a feasible assignment for the rates and window durations, we check whether this returned assignment satisfies the remaining condition for schedulability, i.e., Proposition~\ref{prop:lo_feas}.

\begin{definition}[Multi-rate convex optimization]
\label{def:mr_opt}
Suppose $k_i = i$ for each HI-task $\tau_i \in \tau_H$.
\begin{equation*}
\begin{split}
& minimize \hspace{10pt} \sum_{\tau_i\in\tau_H} \theta_i^L\\
subject~to,\\
& Propositions~\ref{prop:lo_sched}-\ref{prop:hi_feas}\\
& Theorems~\ref{thm:co_hi_mode}~and~\ref{thm:transition_hi_mode} \\
& \forall i: 1 \leq i \leq n_H, \forall j: 1 \leq j \leq n_H, \hspace{5pt} \theta_{i,j}^H  \leq 1\\
& \forall i: 1 \leq i \leq n_H, \hspace{5pt} \theta_{i}^H  \leq 1\\
& \forall j: 1 \leq j \leq n_H, \hspace{5pt} w_j \geq 0
\end{split}
\end{equation*}
\end{definition}

It is easy to see that \emph{SOMA is a Multi-Rate MC-Correct} assignment strategy. Any rate and window assignment returned by SOMA satisfies Propositions~\ref{prop:lo_sched}-\ref{prop:hi_feas} and Theorems~\ref{thm:co_hi_mode} and~\ref{thm:transition_hi_mode} because of the constraints of the optimization problem as well as the check in Step $5$ of Algorithm~\ref{algo:1}.
The following lemma shows that SOMA has a speed-up bound of $4/3$.

\begin{lemma}
Algorithm SOMA has a speed-up bound of $4/3$.
\label{cor:soma_speed_up}
\end{lemma}
\begin{proof} 
From Lemma~\ref{lemma:speed-up} we know that any algorithm that always return some feasible assignment if there exists an assignment satisfying the following three conditions has a speed-up bound of $4/3$.
\begin{enumerate}
\item $\forall i: 1\leq i \leq n_H, \forall j: 1\leq j \leq n_H, \theta_{i,j}^H = \theta_i^H$,
\item $\forall j: 1 \leq j \leq n_H, w_j = 0$, and
\item Propositions~\ref{prop:lo_sched}--\ref{prop:hi_feas} and Theorems~\ref{thm:co_hi_mode} and~\ref{thm:transition_hi_mode}.
\end{enumerate}

To prove this lemma, it is then sufficient to show that such an assignment is considered by SOMA. Observe that the boundary constraints on execution rates and window durations in the optimization problem allow an assignment that satisfies the above three conditions.

The constraints of the optimization problem exactly encode Propositions~\ref{prop:lo_sched}--~\ref{prop:hi_feas} and Theorems~\ref{thm:co_hi_mode} and~\ref{thm:transition_hi_mode} of the schedulability test. 

Now we show that for an assignment satisfying the above conditions, the earliest completion window parameter $k_i$ is irrelevant. 

In Theorem~\ref{thm:co_hi_mode}, Equation~\eqref{eqn:co_nwc1} is subsumed by Equation~\eqref{eqn:co_nwc2}. That is, when condition 1 holds $\forall j: k_i \leq j \leq n_H, \theta_i^L \leq \theta_{i,j}^H$ is subsumed by $\theta_i^L \leq \theta_i^H$.

Also, Equation~\eqref{eqn:co_wc} reduces to $\theta_i^H (T_i - C_i/\theta_i^L) \geq C_i^H-C_i^L$. Substituting $R_i = \theta_i^H$ (by condition 1) and condition 2 in Equation~\eqref{eqn:co_wc} we get,
\begin{align*}
\sum_{j: 1 \leq j < k_i} \theta_{i,j}^H \times w_j + R_i \times \left (T_i - C_i^L/\theta_i^L - \sum_{j: 1 \leq j < k_i} w_j \right) &\geq C_i^H - C_i^L\\
\Rightarrow \theta_i^H (T_i - C_i/\theta_i^L) &\geq C_i^H-C_i^L
\end{align*}
Therefore, Theorem~\ref{thm:co_hi_mode} is independent of $k_i$ values. 

Similarly, in Theorem~\ref{thm:transition_hi_mode}, Equation~\eqref{eqn:transition3} is subsumed by Equation~\eqref{eqn:transition4}. That is, $\forall j: k_i \leq j \leq n_H, \theta_{i,j}^H \geq u_i^H$ becomes $\theta_i^H  \geq u_i^H$ when condition 1 holds.

Equation~\eqref{eqn:transition2}, $\forall j: 1\leq j < k_i, \theta_{i,j}^H \leq \theta_{i,j+1}^H$,  is satisfied because all the transition window rates for each task are identical (i.e., condition 1), and Equation~\eqref{eqn:transition1} is also satisfied because both the sides equate to $0$ (since $w_j=0$, by condition 2). Thus, even this theorem is independent of $k_i$ values. 

Hence the $k_i$ values are irrelevant for an assignment that satisfies the above conditions. Thus, any $k_i$ values, including the ones fixed by SOMA, can be used in the optimization problem.

Thus, the optimization problem is guaranteed to find the assignment that satisfies the above three conditions. This combined with the check for Proposition~\ref{prop:lo_feas} in Step~$5$ of SOMA proves the lemma.\qed
\end{proof}

Algorithm SOMA has a speed-up bound of $4/3$ by Lemma~\ref{cor:soma_speed_up}, and by the argument that no non-clairvoyant algorithm can have a speed-up bound smaller than $4/3$ proves the speed-up optimality of SOMA.

\paragraph{Algorithm Complexity.} SOMA can determine the values for $k_i$s in linearithmic time ($O(n_H log (n_H))$). The optimization has $n_H^2+3n_H$ real variables, $n_H$ integer variables and $9 \times n_H$ constraints. Hence, the number of variables and constraints is polynomially bounded in the number of HI-tasks in the system. The complexity of SOMA is therefore bounded by the complexity of solving the convex optimization problem. By replacing $C_i^L/\theta_i^L$ in the optimization problem with a variable $x_i$, it is easy to see that the problem reduces to a convex optimization problem with objective of the form $\sum_i 1/x_i$ and all linear constraints. This is one of the simplest convex optimization problems, and we plan to investigate a polynomial time algorithm for solving it in future work.

\paragraph{Theoretical performance and experimental evaluation.} Scheduling algorithms are evaluated based on either analytical performance bound or experimental evaluation. Scheduling algorithms with good theoretical performance such as low speed-up bound and good performance in experimental evaluation are usually preferred. Speed-up bound is a good metric to compare the worst-case performance of different scheduling algorithms. For scheduling algorithms such as partitioned scheduling algorithms that are based on heuristics, it is hard to determine these speed-up bounds. The performance of these algorithms is often evaluated based on experimental evaluation. 

The dual-rate fluid model is shown to be speed-up optimal~\cite{baruah_mcf} and has good performance in experimental evaluation for multiprocessor task systems~\cite{lee_mcfluid}. Though the dual-rate model is optimal in terms of speed-up bound, it is not optimal in terms of schedulability~\cite{ramanathan_wmc}. There are several feasible task systems that are deemed unschedulable by MC-Fluid~\cite{lee_mcfluid} - the optimal dual-rate assignment algorithm. Thus, we propose SOMA the multi-rate and window assignment algorithm which has an optimal speed-up bound and performs better than MC-Fluid in experimental evaluation. In fact, our experimental evaluation together with Lemma 3 also shows that SOMA strictly dominates MC Fluid in terms of schedulability. That is, all task systems deemed schedulable by MC Fluid are also deemed so by SOMA, and there exist task systems that are not schedulable by MC Fluid but are deemed so by SOMA.

To provide some insights on the complexity of these two rate assignment algorithms - SOMA and MC-Fluid - we compare their offline and runtime complexity. The rate assignment in both these algorithms is done offline when determining the execution rates. The offline complexity of MC-Fluid rate assignment algorithm is $O(n_H^2)$. Whereas, the offline complexity of SOMA is larger compared to MC-Fluid because the complexity of SOMA is bounded by the complexity of solving the optimization problem defined in Definition~\ref{def:mr_opt}. For details on the complexity of SOMA refer to the complexity discussion above. During runtime, the scheduler executes the task with the assigned execution rates. The runtime scheduling mechanism of SOMA is quite similar to MC-Fluid because the $n_H$ window durations that are determined offline remain fixed during runtime. The runtime complexity of SOMA is not significantly high compared to dual-rate model because the only additional overhead incurred is in the transition period ($n_H$ windows) when the mode switch is triggered. Essentially, upon mode switch in the dual-rate model, HI-tasks switch their execution rate at most once. Whereas, in the case of multi-rate model, each HI-task switches its execution rate at most $n_H+1$ times.
\section{Experiments and Results}
\label{sec:experiments}

In this section we evaluate the schedulability performance of SOMA and compare it with other scheduling algorithms that have known speed-up bounds. These include GLO-EDF\textunderscore VD~\cite{li_globalmc}, PAR-EDF\textunderscore VD~\cite{baruah_multicoremc} and MC-Fluid~\cite{lee_mcfluid}. We compare with only these algorithms because our focus is on designing an algorithm that has a worst-case performance guarantee. We do not compare with the dual-rate strategy MCF~\cite{baruah_mcf} because MC-Fluid is known to be dual-rate optimal.

\paragraph{Task set generation.} Our experiments are carried out using the task set generator proposed in~\cite{ramanathan_waters}. The task set parameters used in our generator are described as follows:
 
\begin{enumerate}
\item $m \in$\{$2,4,8$\} denotes the number of cores.
\item $u_{min}$ ($= 0.001$) and $u_{max}$ ($= 1.0$) denote the minimum and maximum individual task utilization respectively.
\item $U_B = max (U_H^H, U_H^L+U_L^L)/m$ denotes the normalized system utilization in both LO- and HI-modes. We consider $U_B \in[0.50, 0.55, \ldots, 1.0]$.
\item $U_H^H/m\in[0.1, 0.15, \ldots, 1.0]$ denotes the normalized system utilization in HI-mode.
\item $U_H^L/m\in[0.05, 0.10, \ldots, U_H^H/m]$ denotes the normalized system utilization HI-tasks in LO-mode.
\item $U_L^L/m\in[0.05, 0.10, \ldots, 1-U_H^L/m]$ denotes the normalized system utilization of LO-tasks in LO-mode.
\item Total number of tasks is lower bounded by $m+1$ and upper bounded by $10*m$.
\item Total number of HI-tasks in the system is lower bounded by $m+1$ and upper bounded by $3*m$. Note that the performance of SOMA is identical to MC-Fluid for values of $n_H \leq m$. This is expected because if such a task set is dual-rate infeasible, it means that even allocating the maximum rate of $1$ to each task in the HI-mode is not sufficient.
\item $P_H \in [0.1, 0.2, \ldots, 0.9]$ denotes the percentage of HI-tasks in the system.
\item $T_i$, the period of task $\tau_i$ is drawn uniformly at random from $[5, 100]$.
\item Task utilizations $u_i^L$  and $u_i^H$ are determined using techinques \emph{MRandFixedSum}~\cite{emberson_mrand} and \emph{BoundedUniform}~\cite{ramanathan_waters}
\item The execution requirements $C_i^L$ and $C_i^H$ are defined as $u_i^L \times T_i$ and $u_i^H \times T_i$ respectively.
\end{enumerate}

Using the above procedure we generate atleast $1000$ task sets for each value of $U_B$ and $m$. For each such successfully generated task set we evaluate the performance of the SOMA with other multi-core MC scheduling algorithms with known speed-up bound.

\begin{figure}
	    \centering
	    \subfigure[$m=2$]{
            	\includegraphics[width=0.75\textwidth,height=5cm]{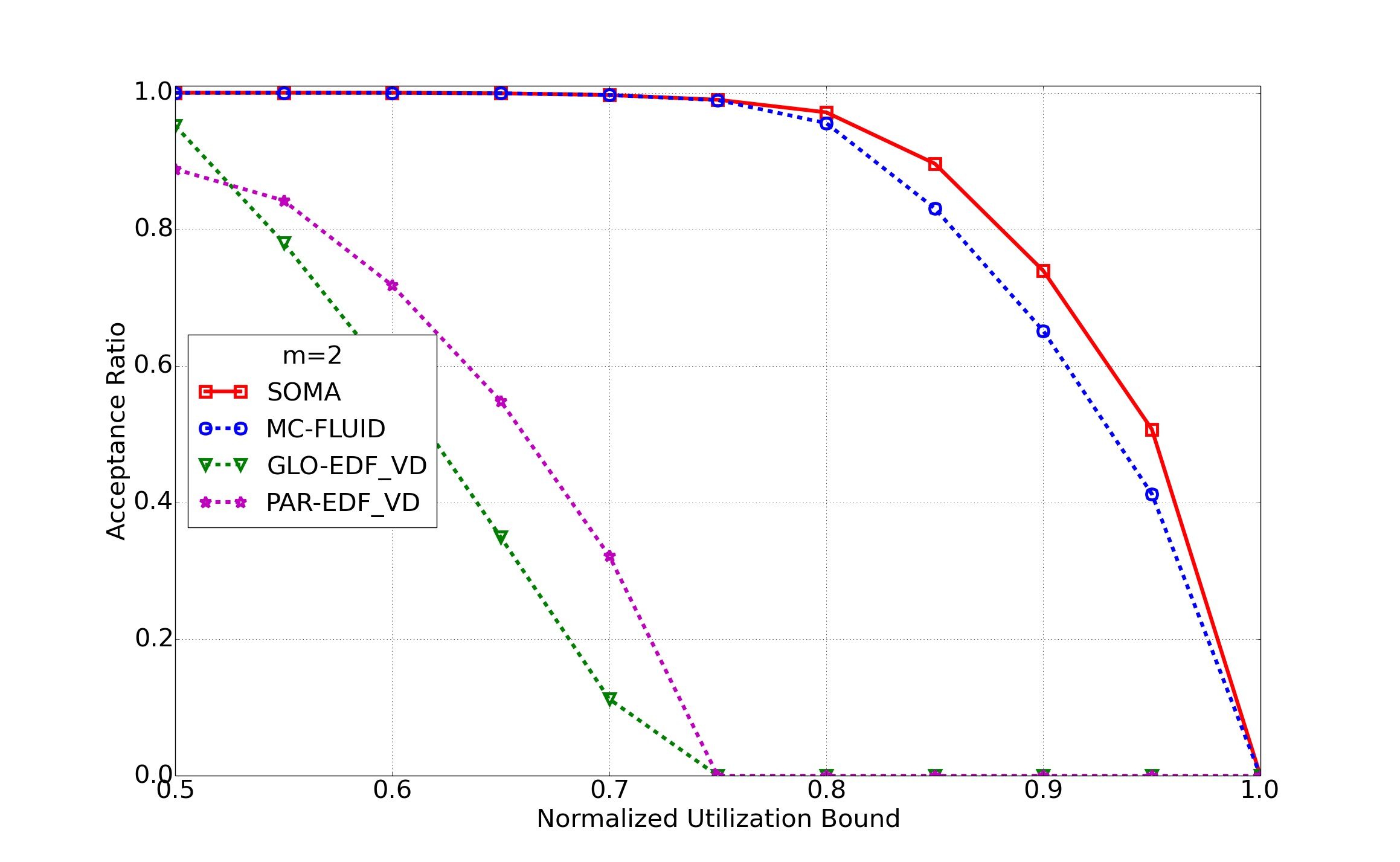}
            	\label{fig:sched_m2}
	     }
	     \subfigure[$m=4$]{
            	\includegraphics[width=0.75\textwidth,height=5cm]{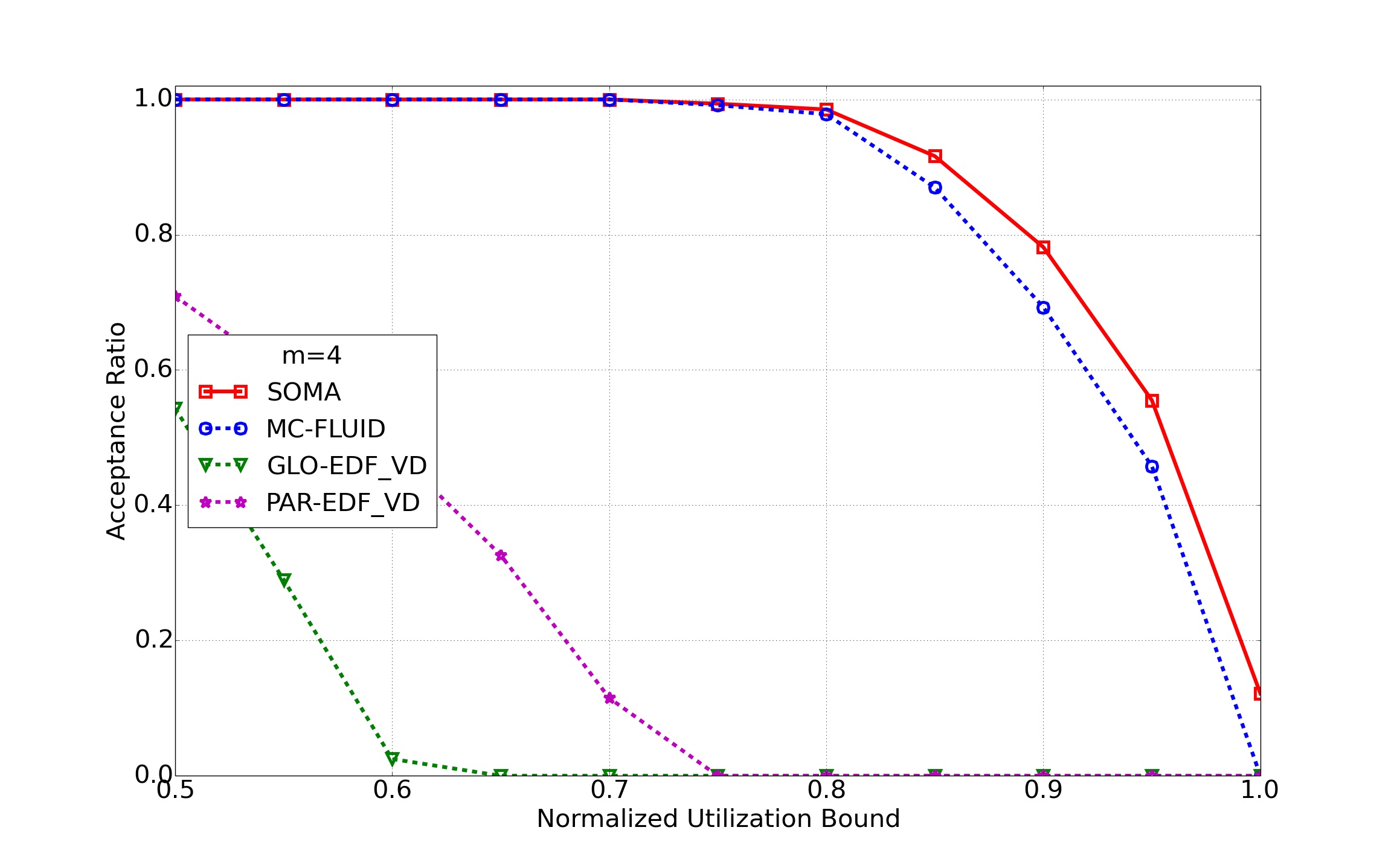}
            	\label{fig:sched_m4}
	     }
	     \subfigure[$m=8$]{
	     \includegraphics[width=0.75\textwidth,height=5cm]{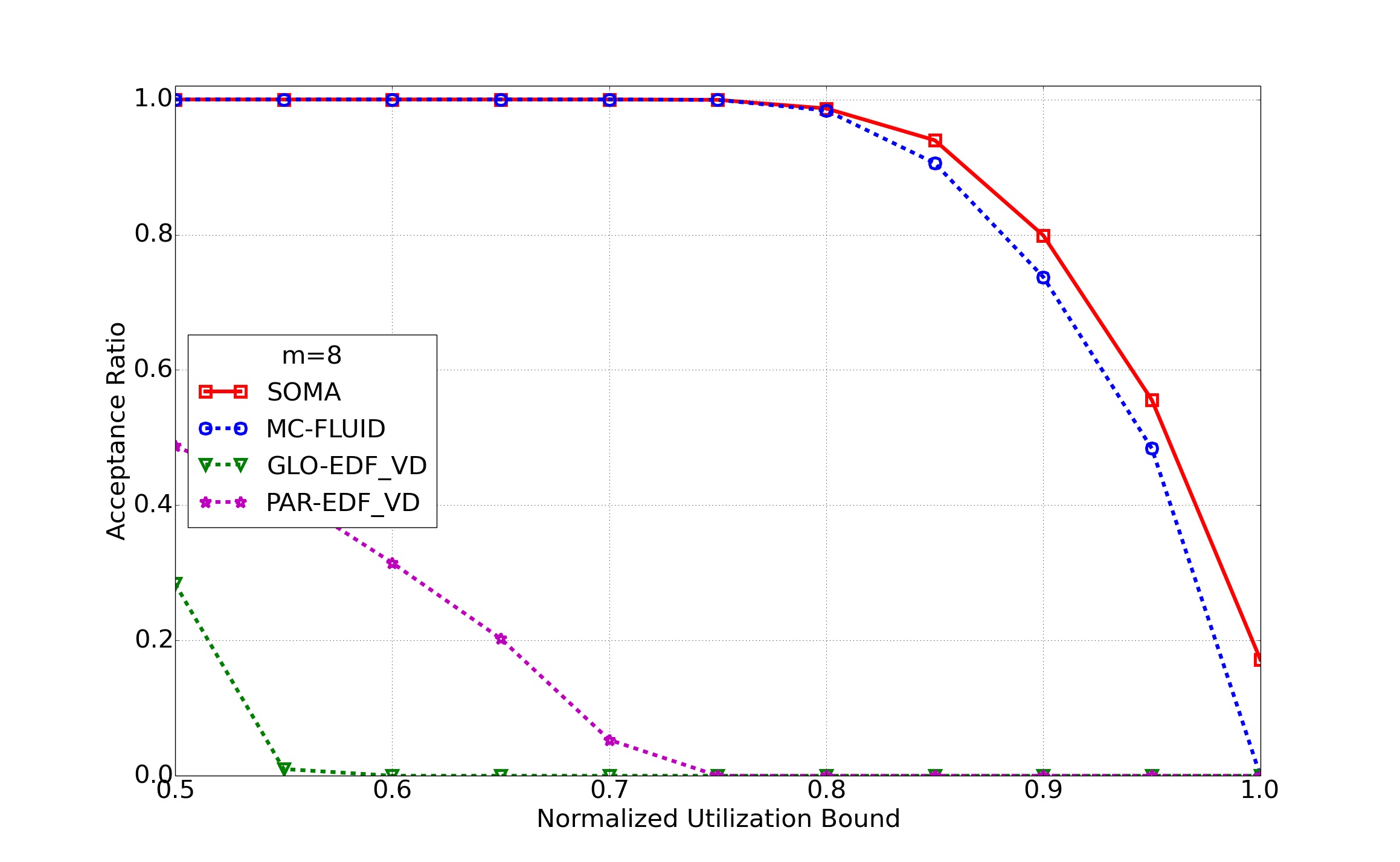}
	     \label{fig:sched_m8}
	     	     }
\caption{Comparison of acceptance ratio for varying $m$}
\label{fig:schedulability}
\end{figure}

\paragraph{Results.} Figure~\ref{fig:schedulability} shows the overall performance of the algorithms where we plot the acceptance ratio, i.e., fraction of generated task sets that are deemed schedulable, versus normalized utilization $U_B$ for varying $m\in[2,4,8]$. Each data point in the plot corresponds to $1000$ task sets. As can be seen, SOMA clearly dominates all the other algorithms. It significantly outperforms GLO-EDF\textunderscore VD and PAR-EDF\textunderscore VD, and also has better performance than MC-Fluid. Further, the performance gap between SOMA and MC-Fluid increases with $m$ at high $U_B$ values mainly because the number of tasks in the task sets are also higher.

\begin{figure}
\centering
\includegraphics[width=0.75\textwidth,height=5cm]{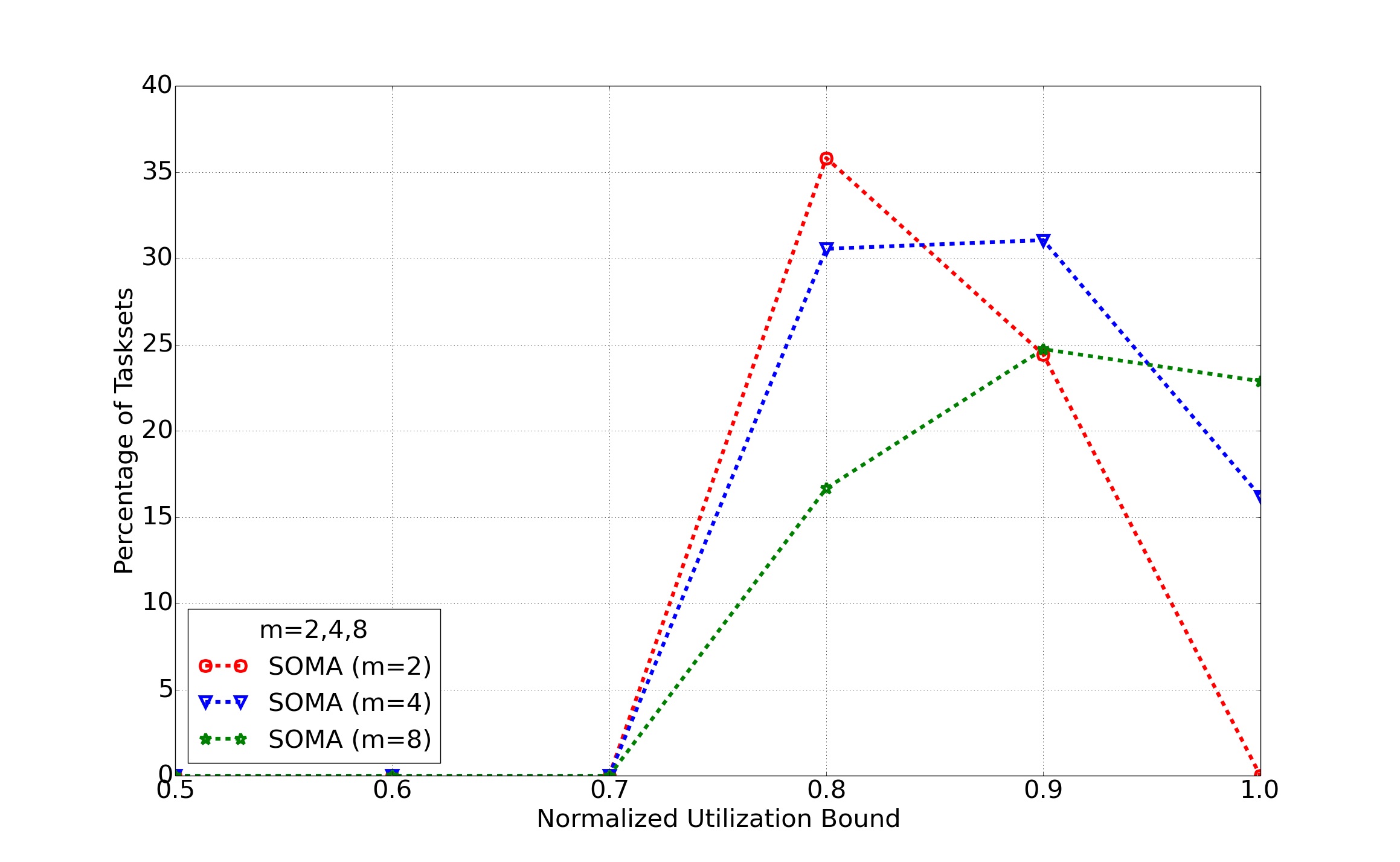}
\caption{Percentage of dual-rate infeasible task sets that are schedulable under multi-rate}
\label{fig:bruteforce}
\end{figure}

In the remainder of this section we take a closer look at the performance of SOMA in comparison to MC-Fluid.
Since an exact feasible test for MC task systems is not known ($U_B \leq 1$ is only necessary but not sufficient), it is difficult to exactly characterize the performance gap between SOMA and MC-Fluid. Nevertheless, to provide some insights into the significance of this gap, we show the percentage of \emph{dual-rate infeasible task sets} that are deemed to be schedulable under SOMA in Figure~\ref{fig:bruteforce}. Due to speed-up optimality of MC-Fluid (speed-up of $4/3$), there are very few feasible task sets with $U_B \leq 0.75$ that are unschedulable under MC-Fluid. There is no difference in performance between SOMA and MC-Fluid for $U_B \leq 0.75$, because the two algorithms have optimal performance for $U_B\leq0.75$. 
As expected, we can see that SOMA performs better for $U_B \geq 0.80$. For $m=2$, as much as $35.8\%$ more task sets are deemed to be schedulable under SOMA. Thus, there are significant number of task sets that are schedulable under SOMA but not under MC-Fluid (about $358$ out of $1000$ when $U_B = 0.80$). We believe this is a significant improvement over the state-of-art algorithm with known speed-up bound. The reason why this performance gap is not explicitly visible in Figure~\ref{fig:schedulability} is because we plot the percentage of schedulable task sets inclusive of both the dual-rate feasible and dual-rate infeasible task sets. Whereas, in Figure~\ref{fig:bruteforce} we only consider the dual-rate infeasible task sets. For $m=2$, when $U_B = 0.80$ the acceptance ratio of SOMA and MC-Fluid in Figure~\ref{fig:schedulability} is $97.1\%$ and $95.5\%$ respectively. The performance gap between SOMA and MC-Fluid seems to be only $1.6\%$ ($97.1\%-95.5\%$). But, $1.6\%$ out of the $4.5\%$ ($100\%-95.5\%$) dual-rate infeasible task sets are schedulable by SOMA. This translates to a performance improvement of $35.5\%$ ($1.6/4.5 \times 100\%$) over the MC-Fluid which is consistent with Figure~\ref{fig:bruteforce}. It is also important to note that since exact feasibility tests are not known for mixed-criticality systems, the task sets generated in the evaluation are not necessarily all feasible. In fact, it is reasonable to expect that the number of feasible task sets among those generated in the evaluation would decrease with increasing $m$ and utilization. For $m=4$ and $m=8$, the performance gap between SOMA and MC-Fluid increases gradually as $U_B$ increases. This is expected because the number of HI-tasks in the task sets are also higher for same value of $U_B$.

\begin{figure}[h!]
	    \centering
	    \subfigure[$m=2$]{
           	\includegraphics[width=0.75\textwidth,height=5cm]{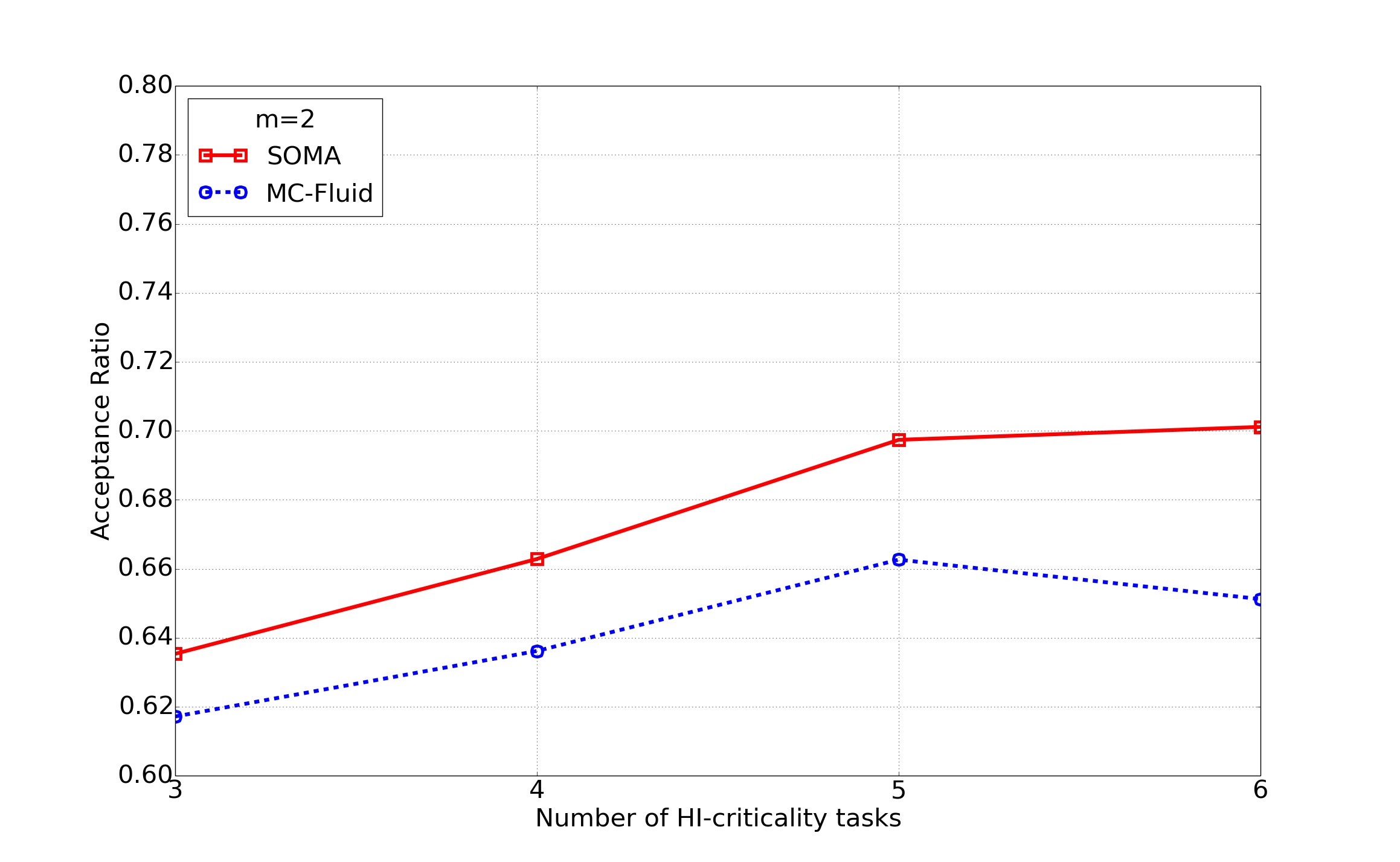}
           	\label{fig:hitasks}
	     }
	     \subfigure[$m=4$]{
            	\includegraphics[width=0.75\textwidth,height=5cm]{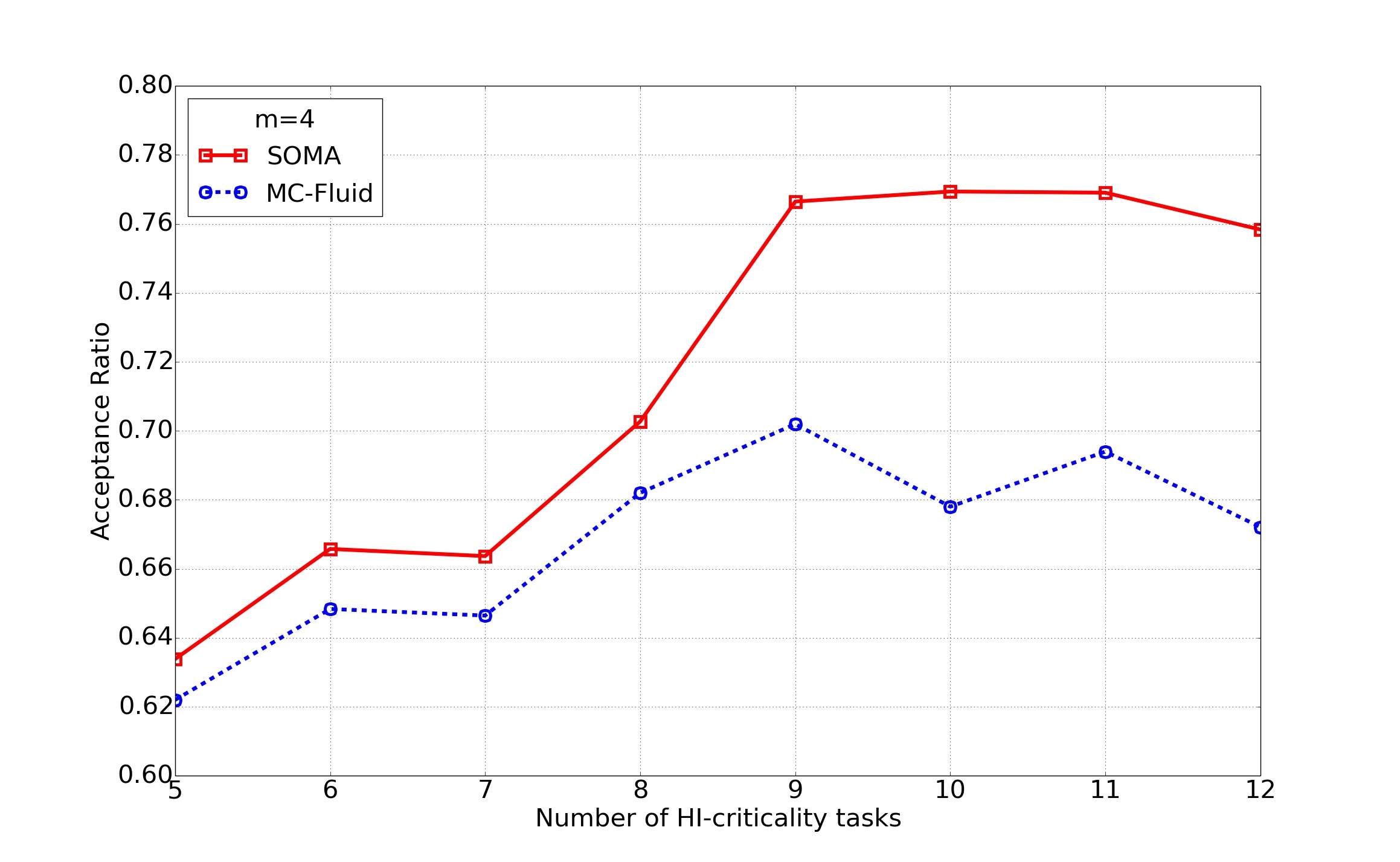}
           	\label{fig:hitasks2}
	     }
	     \subfigure[$m=8$]{
            	\includegraphics[width=0.75\textwidth,height=5cm]{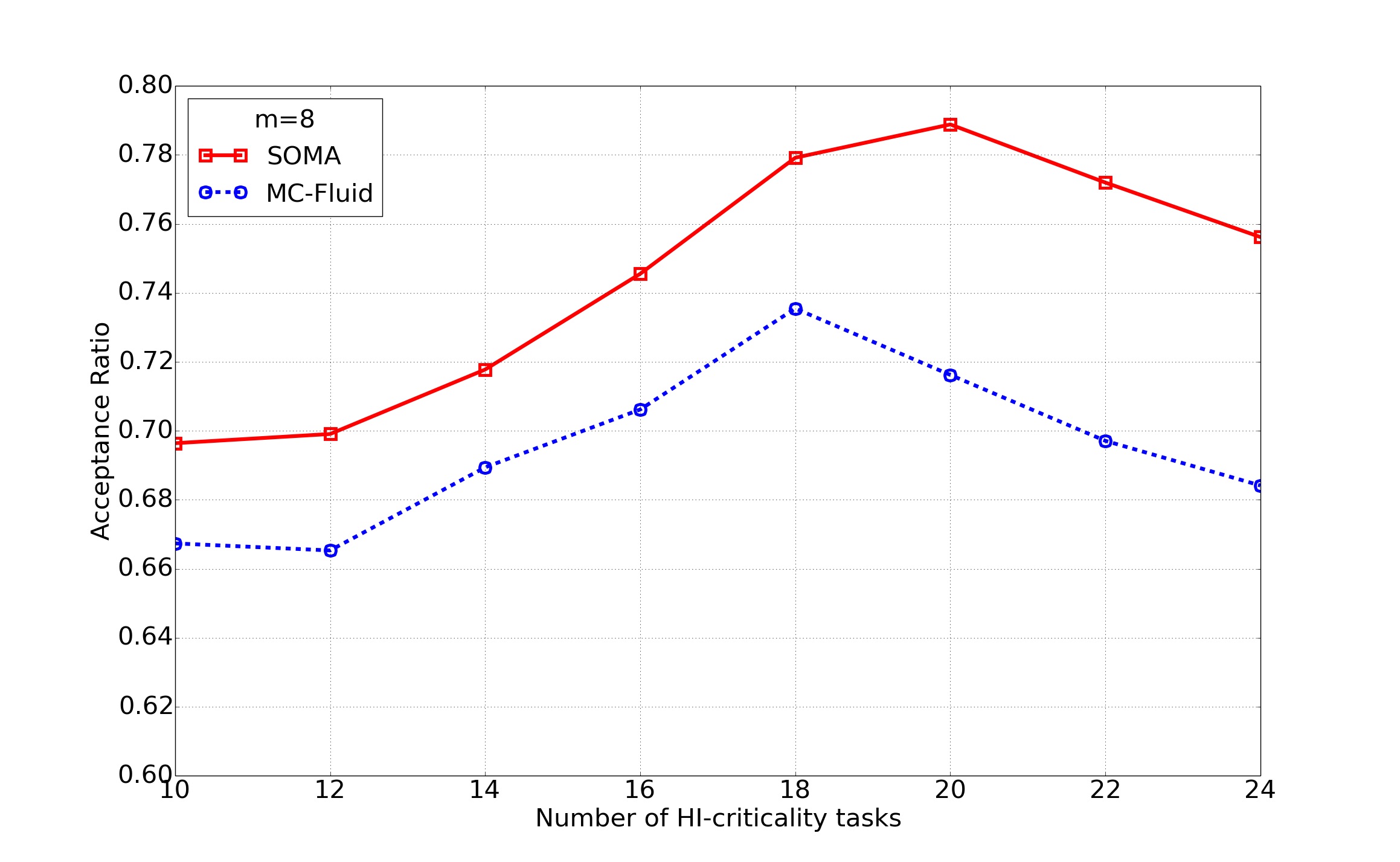}
           	\label{fig:hitasks8}
	     }
\caption{Weighted Acceptance ratio with varying $n_H$}
\label{fig:acceptance_ratio_nH}
\end{figure}

Next, we present the weighted acceptance ratios for varying values of task set parameters $P_H$ and $u_{max}$. Weighted Acceptance Ratio for an algorithm is defined as $WAR(\mathbb{S}) = \frac{\underset{U_B\in\mathbb{S}}{\sum}(AR(U_B) X U_B)}{\underset{U_B\in\mathbb{S}}{\sum} U_B}$, where $\mathbb{S}$ is the set of $U_B$ values and AR($U_B$) is the acceptance ratio of that algorithm for a specific value of $U_B$. This metric essentially gives more importance to task sets with higher $U_B$ values, consistent with the fact that such task sets are in general harder to schedule. Task sets are generated using the procedure described in Sect.~\ref{sec:experiments}. The results of these task sets are combined based on their normalized utilization $U_B$.

An important factor in MC fluid scheduling is the number of HI-tasks ($n_H$). Therefore, in Figure~\ref{fig:hitasks} we compare the weighted acceptance ratios of SOMA and MC-Fluid for varying $n_H$ and $m$. As $n_H$ increases the performance gap between SOMA and MC-Fluid also increases. This is expected because for larger $n_H$ at the same $U_B$ value, there is more flexibility in assigning multiple execution rates in the HI-mode.

\begin{figure}
	    \centering
	    \subfigure[Varying percentage of HI-tasks($P_H$)]{
            	\includegraphics[width=0.75\textwidth,height=5cm]{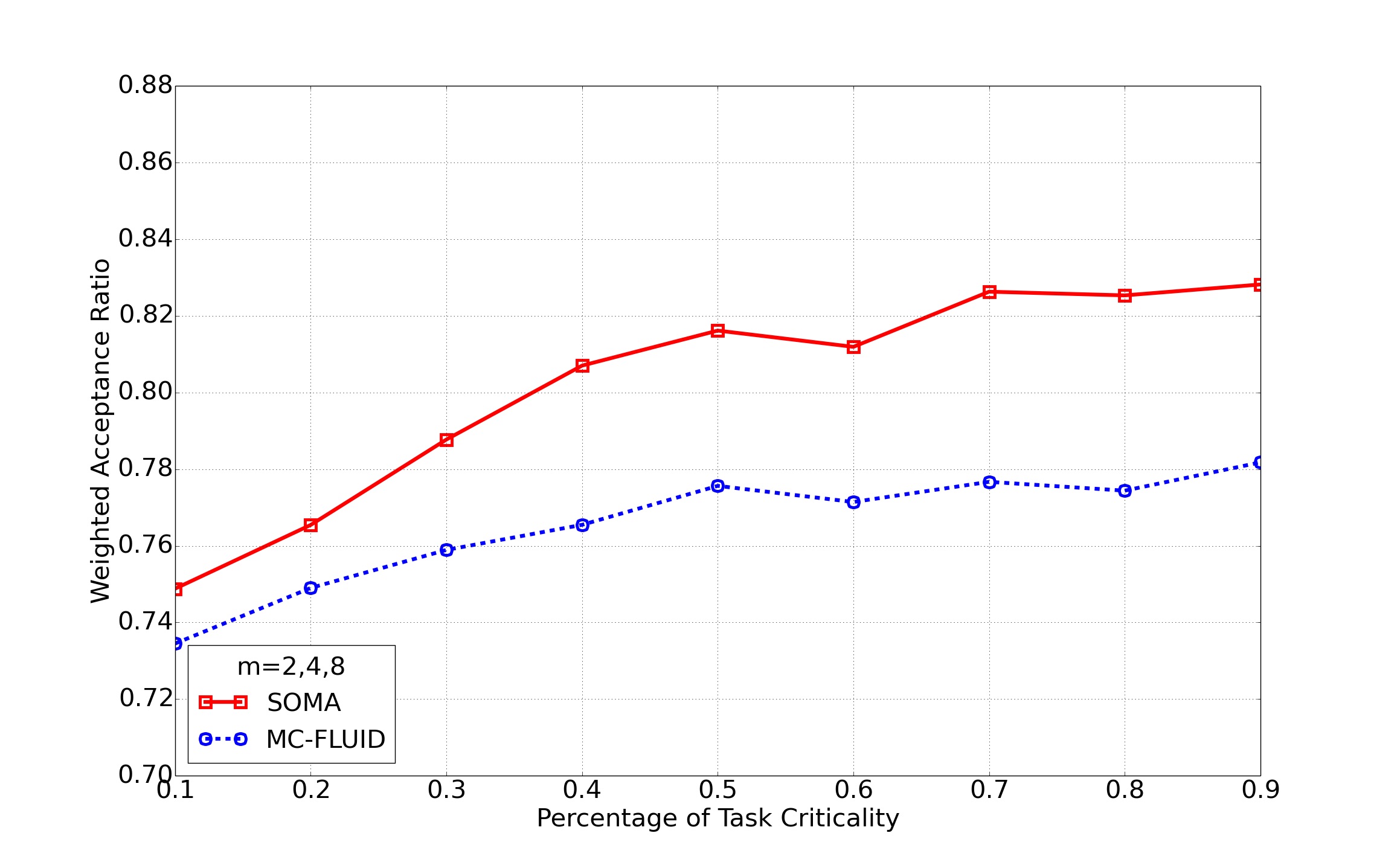}
            	\label{fig:probability}
	     }
	     \subfigure[Varying maximum HI-utilization of HI-tasks ($max\{u_i^H\}$)]{
            	\includegraphics[width=0.75\textwidth,height=5cm]{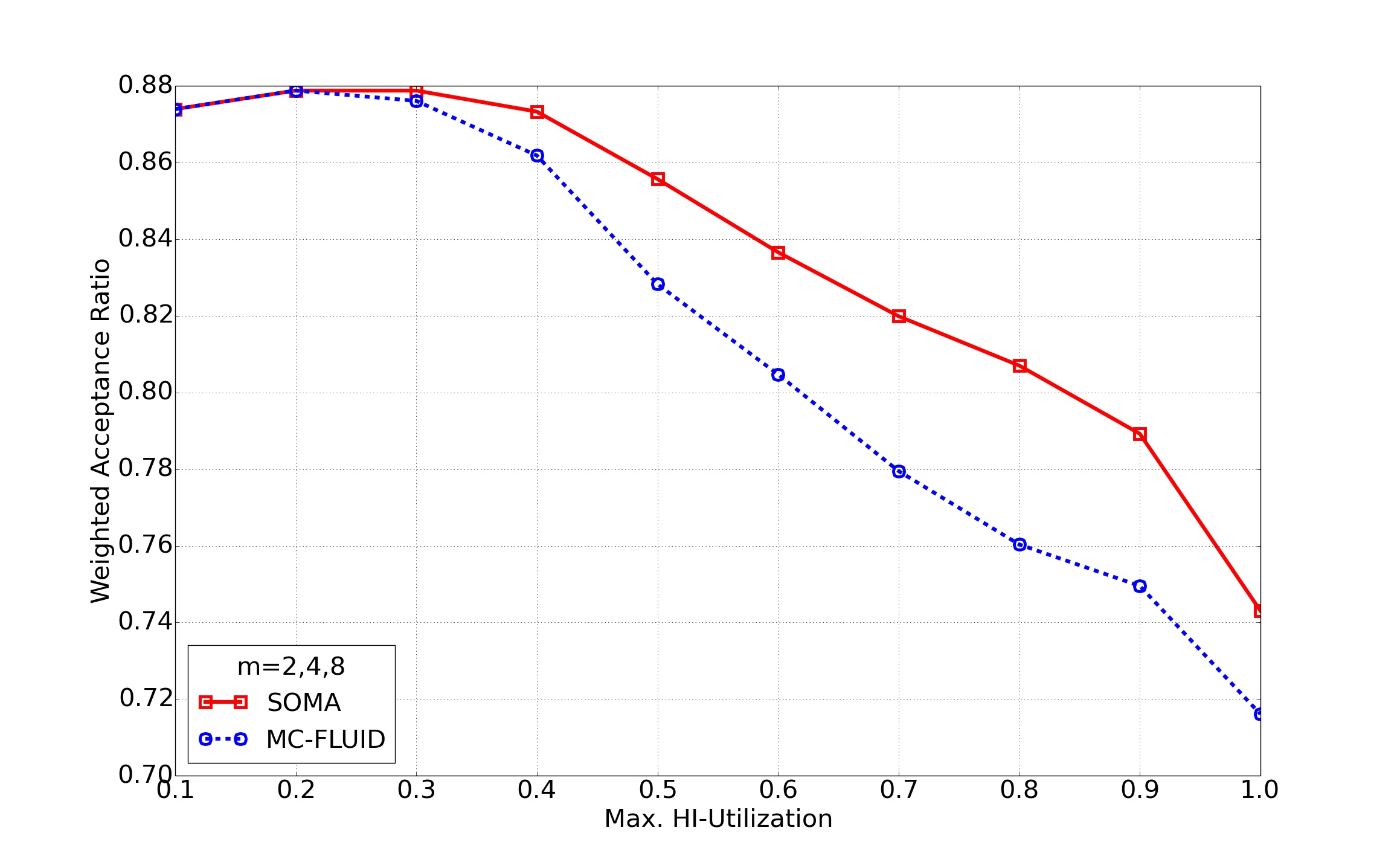}
            	\label{fig:utilization}
	     }
\caption{Comparison of weighted acceptance ratio}
\label{fig:weighted_schedulability}
\end{figure}
In Figure~\ref{fig:probability} we compare the weighted acceptance ratios of SOMA and MC-Fluid for varying $P_H$ values. The performance gap widens as the number of HI-tasks in a task set increases, indicating that SOMA performs better when there are more number of HI-tasks in the system. This is as expected because the main benefit of assigning multiple rates to HI-tasks in the HI-mode is that the total rate for these tasks in the LO-mode can be reduced. This reduction is useful only when there are LO-tasks that can benefit from it. At high $P_H$ values, there are very few LO-tasks and as a result, all these tasks will be benefitted.

In Figure~\ref{fig:utilization} we compare the weighted acceptance ratios of SOMA and MC-Fluid for varying $max\{u_i^H\}$ values. The performance of both algorithms decrease with increasing $max\{u_i^H\}$, this is reasonable because typically the number of tasks in each task set reduces with increasing $max\{u_i^H\}$ and as shown in Sect.~\ref{sec:experiments}, SOMA tends to perform better with increasing $n_H$. When $max\{u_i^H\}$ is large, there is a significant performance gap between SOMA and MC-Fluid. This is because in such task systems there are typically more number of heavy utilization tasks. As a result, the minimum required HI-mode execution rate under MC-Fluid is high and hence it does not have much flexibility. However, under SOMA, some execution rates in the HI-mode can be much lower than $u_i^H$ as long as the average rate over all the transition windows is reasonably high. This flexibility in rate assignment is a key property of SOMA and its benefit is clearly visible for task sets with high $max\{u_i^H\}$ values.

In Figure~\ref{fig:system_utilization} we compare the weighted acceptance ratios of SOMA and MC-Fluid for varying $U_L^L/m$ values with $U_B = 0.95$. The performance gap widens as $U_L^L/m$ increases, indicating that SOMA performs better when the system utilization of LO-tasks is high. This is reasonable because SOMA tries to minimize the total LO-mode execution rate of HI-tasks and as a result, this reduction becomes more useful as $U_L^L/m$ increases.
\begin{figure}
\centering
\includegraphics[width=0.75\textwidth,height=5cm]{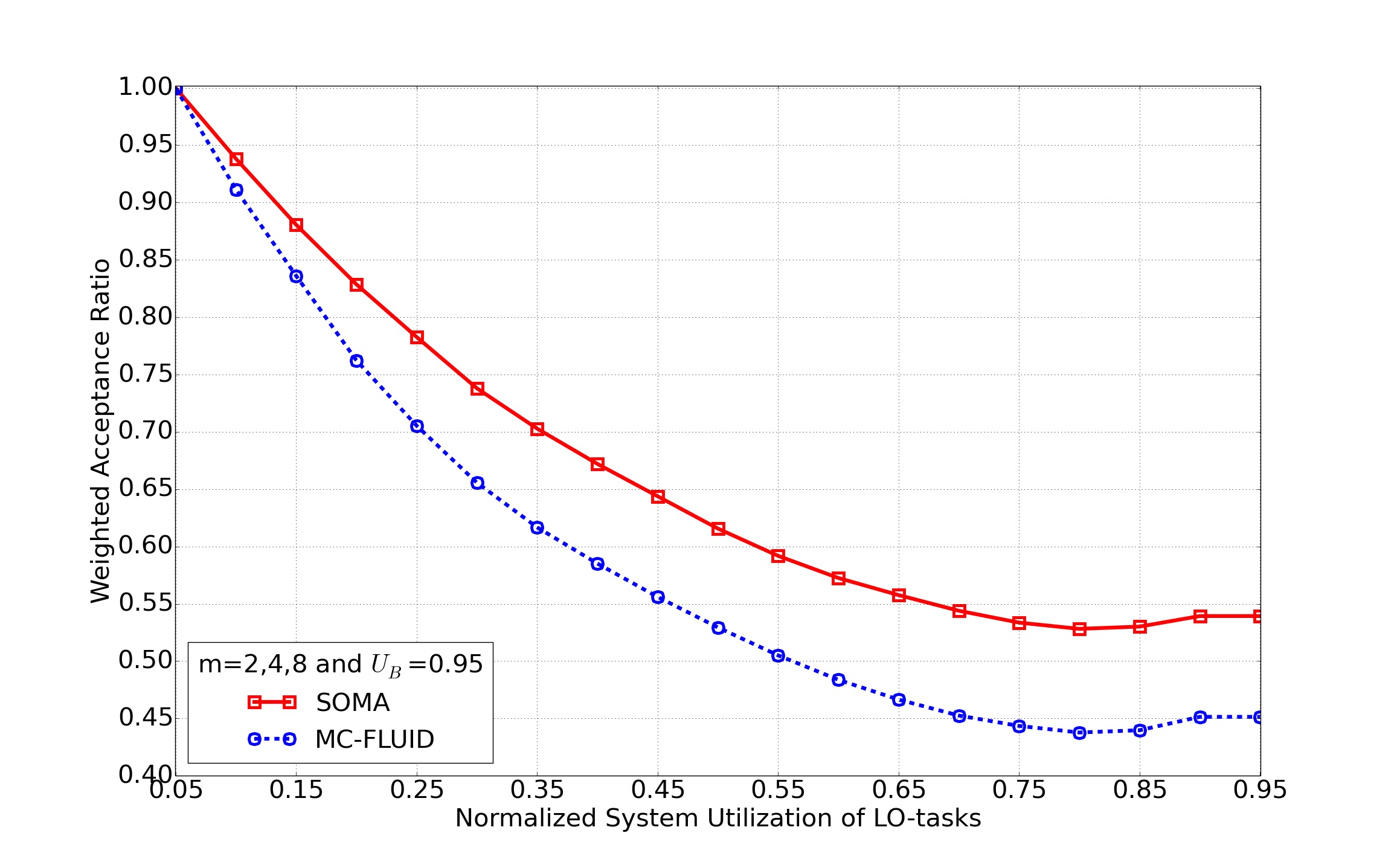}
\caption{Varying normalized system utilization of LO-tasks ($U_L^L/m$)}
\label{fig:system_utilization}
\end{figure}

The presented experiments show that SOMA outperforms all the existing algorithms with known speed-up bounds. The overall performance gap between SOMA and MC-Fluid increases with utilization. Further, we have also shown that this gap is significant in terms of the percentage of dual-rate infeasible task sets that SOMA can schedule, and that the gap widens with increasing number of HI-tasks in the system.
\section{Summary}
\label{sec:summary}

Scheduling algorithms with good speed-up bounds as well as good performance in schedulability experiments are highly desirable. In this paper, we addressed the shortcoming in the schedulability performance of the speed-up optimal dual-rate fluid model by proposing the multi-rate fluid model. We derived a sufficient schedulability test for the multi-rate model and also showed that it dominates the dual-rate model.
Finally, we presented a speed-up optimal algorithm to compute the execution rates and window durations for the multi-rate model, and showed through experiments that it outperforms all the other MC scheduling algorithms with known speed-up bounds. 

As part of future work, we plan to derive an optimal strategy for computing the execution rates and window durations in the multi-rate model that runs in polynomial time.
Extending the proposed multi-rate model to systems with multiple criticality levels is another research direction that we plan to explore.
\section*{Acknowledgement}

We would like to thank Jaewoo Lee for providing motivation for this work through discussions on the sub-optimality of dual-rate fluid scheduling.

\bibliographystyle{spmpsci}      
\bibliography{all}   

%
%
\end{document}